\documentclass[aps,nofootinbib,pra,reprint,superscriptaddress]{revtex4-2}
\usepackage[ruled,vlined]{algorithm2e}
\usepackage{amsfonts,amssymb,amsthm,bm,etoolbox,graphicx,hyperref,moresize,physics,qcircuit,setspace,tikz,upgreek}
\usetikzlibrary{shapes.geometric}
\theoremstyle{definition}

\newtheorem{definition}{Definition}
\theoremstyle{plain}
\newtheorem{lemma}{Lemma}
\newtheorem{theorem}{Theorem}
\newtheorem{corollary}{Corollary}
\newenvironment{proofs}{\proof}{\endproof}

\newcommand{\ce}{\ensuremath{\mathrm{e}}}
\newcommand{\ci}{\ensuremath{\mathrm{i}}}
\newcommand{\cpi}{\ensuremath{\uppi}}

\definecolor{JLHgreen}{RGB}{50,100,50}

\usepackage[normalem]{ulem}

\begin{document}
\title{Interpretable Quantum Advantage in Neural Sequence Learning}
\author{Eric R.\ Anschuetz}
\email{eans@mit.edu}
\affiliation{MIT Center for Theoretical Physics, 77 Massachusetts Avenue, Cambridge, MA 02139, USA}
\author{Hong-Ye Hu}
\affiliation{Department of Physics, University of California San Diego, 9500 Gilman Drive, La Jolla, CA 92093, USA}
\affiliation{Harvard Quantum Initiative, Harvard University, 17 Oxford Street, Cambridge, MA 02138, USA}
\affiliation{Department of Physics, Harvard University, 17 Oxford Street, Cambridge, MA 02138, USA}
\author{Jin-Long Huang}
\affiliation{Department of Physics, University of California San Diego, 9500 Gilman Drive, La Jolla, CA 92093, USA}
\author{Xun Gao}
\email{xungao@g.harvard.edu}
\affiliation{Department of Physics, Harvard University, 17 Oxford Street, Cambridge, MA 02138, USA}
\preprint{MIT-CTP/5456}

\begin{abstract}
    Quantum neural networks have been widely studied in recent years, given their potential practical utility and recent results regarding their ability to efficiently express certain classical data. However, analytic results to date rely on assumptions and arguments from complexity theory. Due to this, there is little intuition as to the source of the expressive power of quantum neural networks or for which classes of classical data any advantage can be reasonably expected to hold. Here, we study the relative expressive power between a broad class of neural network sequence models and a class of recurrent models based on Gaussian operations with non-Gaussian measurements. We explicitly show that quantum contextuality is the source of an unconditional memory separation in the expressivity of the two model classes. Additionally, as we are able to pinpoint quantum contextuality as the source of this separation, we use this intuition to study the relative performance of our introduced model on a standard translation data set exhibiting linguistic contextuality. In doing so, we demonstrate that our introduced quantum models are able to outperform state of the art classical models even in practice.
\end{abstract}

\maketitle

\section{Introduction}\label{sec:introduction}

The field of quantum information processing has reached a watershed in recent years, with the first demonstrations of quantum processors performing tasks on the verge of classical intractability~\cite{arute2019quantum,doi:10.1126/science.abe8770,PhysRevLett.127.180501,PhysRevLett.127.180502,ZHU2022240,hangleiter2022}. Spurred on by these recent experimental developments, there has been a push for finding algorithms that can be performed using either near-term quantum devices, or early error-corrected ones. One of the leading candidates for such algorithms are \emph{quantum machine learning} (QML) algorithms, where training can be offloaded to a classical computer working in conjunction with a quantum computer, potentially minimizing the coherence requirements of the quantum device~\cite{schuld2015introduction,biamonte2017quantum,PhysRevX.8.021050,perdomo2018opportunities,PhysRevLett.121.040502,PhysRevLett.122.040504,havlivcek2019supervised,PhysRevResearch.1.033063}. These algorithms are also motivated by the ability of quantum systems to naturally represent complex probability distributions that are believed to be difficult to represent classically~\cite{doi:10.1126/sciadv.aat9004,coyle2020born,PhysRevResearch.2.033125,sweke2021quantum}, with many proposed architectures for such quantum models~\cite{farhi2018classification,PhysRevX.8.021050,doi:10.1126/sciadv.aav2761,PhysRevA.100.052327}.

However, any proof of advantage in the expressivity of these models over classical models relies on results from computational complexity theory, themselves conditional on complexity theoretic assumptions~\cite{doi:10.1126/sciadv.aat9004,coyle2020born,PhysRevResearch.2.033125,sweke2021quantum,liu2021rigorous}. As the proofs of separation are abstract, it is unclear what realistic classical data sets one should expect a separation to hold in practice. Also, due to the universality of many of these models, they are very likely to be untrainable due to phenomena such as barren plateaus~\cite{mcclean2018barren,cerezo2021cost,marrero2021entanglement,napp2022quantifying} and bad local minima~\cite{anschuetz2021critical,you2021exponentially,anschuetz2022barren} present in their loss landscapes. Because of these concerns, it has become increasingly clear that quantum models should be carefully constructed to fit the task at hand. Above all else, the \emph{interpretability} of any expressivity separation achieved by a QML model has become increasingly important. Interpretability reveals which features of quantum mechanics yield more expressive models compared to classical models and, armed with this knowledge, allows one to find classes of problems where a practical quantum advantage on real data is potentially achievable.

Wishing to construct a model with an interpretable quantum advantage, we here focus on \emph{sequence-to-sequence} learning tasks~\cite{10.5555/2969033.2969173}, and consider a quantization of \emph{linear recurrent neural networks} (LRNNs)~\cite{Hopfield2554}. Classical LRNNs are recurrent neural networks with only linear activation functions. Such models can equivalently be considered a classical dynamical system governed by quadratic Hamiltonian evolution in the canonical variables $\left(\bm{q},\bm{p}\right)$. By lifting these canonical variables to operators $\left(\bm{\hat{q}},\bm{\hat{p}}\right)$ that satisfy the canonical commutation relations (in units where $\hbar=\frac{1}{2}$):
\begin{equation}
    \left[\hat{q}_j,\hat{p}_k\right]=\frac{\ci}{2}\delta_{jk},
\end{equation}
we arrive at a continuous variable (CV) quantum model where time evolution on an eigenstate of the canonical operators is performed under a quadratic Hamiltonian. To measure properties of the state of the system, the most natural choice is to perform \emph{homodyne measurement}; that is, measure linear combinations of the canonical operators $\hat{q}_j$ and $\hat{p}_k$. This yields a quantum generative model where all operations are Gaussian. However, as all operations, initial states, and measurements are Gaussian, there are efficient Wigner function based simulations of sampling from such a system~\cite{PhysRevLett.109.230503}. In other words, such models on $n$ modes are equivalent to deep belief networks~\cite{CHEESEMAN198854}---a class of commonly used classical models---with $2n$ latent variables.

Instead, we extend this model slightly further by allowing for measurements of the canonical operators \emph{modulo $2\cpi$}, beginning in an eigenstate of periodic functions of the canonical operators~\cite{PhysRevA.64.012310,PhysRevA.82.022114}. We call this introduced class of models \emph{contextual recurrent neural networks} (CRNNs). Our main result is that CRNNs are more memory efficient at expressing certain distributions than essentially all trainable classical sequence models, even though CRNNs are not universal for CV quantum computation. Concretely, we show unconditionally that there exists a class of CRNNs with $\operatorname{O}\left(n\right)$ qumodes that can express certain distributions that no ``reasonable'' (which we later describe) classical model is able to represent without an $\operatorname{\Omega}\left(n^2\right)$-dimensional latent space. Though this is only a quadratic separation in memory, the time complexity of inference for classical models is typically superlinear in the model size~\cite{Hopfield2554,10.1162/neco.1997.9.8.1735,cho-etal-2014-learning,10.5555/3295222.3295349}, yielding a superquadratic time separation. As we show a memory (rather than a time) separation, our results also potentially point to a practical \emph{generalization} advantage for CRNNs, as smaller models tend to generalize better than larger models due to formalized versions of Occam's razor~\cite{dziugaite2020revisiting}.
\begin{figure}
    \begin{center}
        \includegraphics[width=\linewidth]{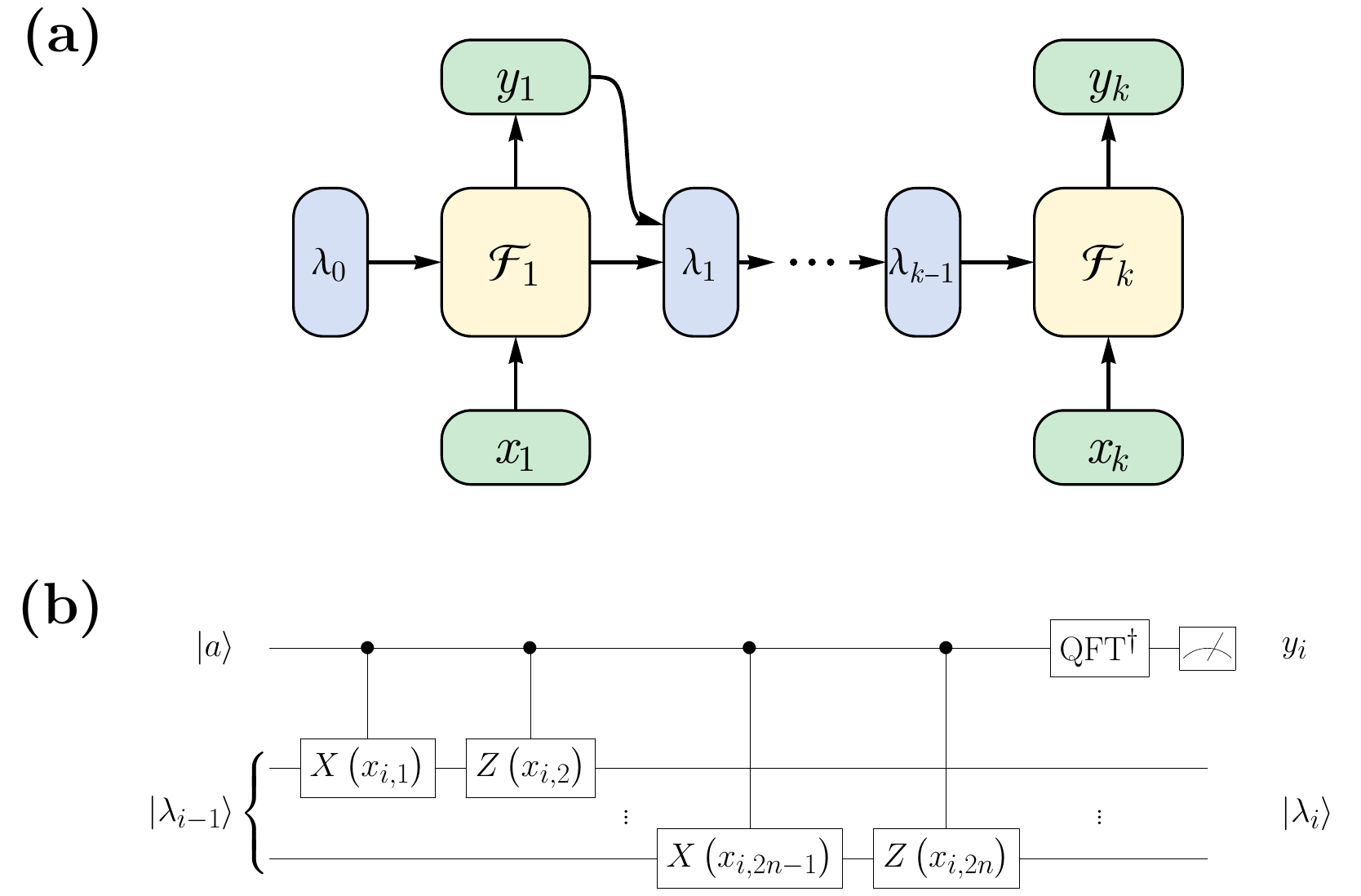}
        \caption{(a) An online neural sequence model. The model autoregressively takes input tokens $\bm{x_i}$ and outputs decoded tokens $\bm{y_i}$ with the map $\mathcal{F}_i$. The model also has an unobserved internal memory with state $\bm{\lambda_{i-1}}\in L$ that $\mathcal{F}_i$ can depend on. When the model is quantized to a CRNN, the $n$-dimensional space of $\lambda_i$ is promoted to the Hilbert space of $n$ qumode states $\ket{\lambda_i}$. (b) An implementation of a phase estimation circuit for CV Pauli operators, which forms the recurrent cell of the CRNN we use to prove our separations. Here, $\ket{a}$ is a fixed ancilla state. Formally, if $\ket{a}$ is a GKP state, this circuit allows for infinite precision measurements. In practice, $\ket{a}$ can be a tensor product of a constant number of qubit $\ket{+}$ states for finite precision phase estimation.\label{fig:online_contextual_model}}
    \end{center}
\end{figure}

Moreover, we are able to show directly that this quantum advantage is due to quantum contextuality~\cite{gleason1975measures,bell1966problem,kochen1975problem,PhysRevLett.65.3373,PhysRevA.82.022114} present in our quantum model. Previously, quantum contextuality was known to be the resource for the expressive power of a certain class of quantized Bayesian networks~\cite{gao2021enhancing}. Our results show that this resource can be used to separate quantum models even from neural networks, which are exponentially more efficient than generic Bayesian networks. Intuitively, quantum contextuality is the statement that quantum measurement results depend on which measurements were previously performed, even if the measurements in question commute. In other words, quantum contextuality is the statement that the measurement of quantum observables cannot be thought of as the revealing of preexisting classical values for the observables. Here, we give a proof of the intuition that reasonable classical models cannot get around the need to ``memorize'' the measurement context of given observables, which is what yields the quadratic memory separation between the quantum and classical models.

Qualitatively, quantum contextuality is similar to the linguistic contextuality present in sentences. Namely, the meaning of a given word in a sentence depends heavily on other words in the sentence, and without this context has no fixed, single meaning. Inspired by this, we also test our constructed model against state of the art classical models on a real-world translation task. In particular, we evaluate the performance of an LRNN~\cite{pmlr-v70-jing17a}, an RNN with gated recurrent units (GRU RNN)~\cite{cho-etal-2014-learning}, a Transformer~\cite{10.5555/3295222.3295349}, a Gaussian model, and our introduced contextual model on a standard Spanish-to-English data set~\cite{spaengdata}. We show that our introduced contextual model achieves better translation performance compared to all other models at each model size we consider. This separation holds even when the online models are constrained to have a similar (and where possible, the same) number of trainable parameters in each recurrent cell.

Our methods provide a novel strategy for designing QML models for near-term devices: through the quantization of simple classical machine learning models with some minimal quantum extension. Though such models are most likely unable to outperform state of the art classical machine learning models on \emph{all} tasks, the intuition gleaned from the simplicity of the quantum models gives guidance as to which problems the quantum models may outperform classical models on. Furthermore, the simplicity of the quantum models may circumvent the recent deluge of untrainability results of general quantum models~\cite{mcclean2018barren,cerezo2021cost,marrero2021entanglement,napp2022quantifying,anschuetz2021critical,you2021exponentially,anschuetz2022barren}. Finally, as such models are restricted in their allowed operations, they are more amenable to implementation on near-term quantum devices than completely generic quantum models.

\section{Classical and Quantum Neural Sequence Learning}\label{sec:background}

\subsection{Classical Sequence Learning}\label{sec:classical_sequence_learning}

\emph{Sequence-to-sequence} or \emph{sequence} learning~\cite{10.5555/2969033.2969173} is the approximation of some given conditional distribution $p\left(\bm{y}\mid\bm{x}\right)$ with a model distribution $q\left(\bm{y}\mid\bm{x}\right)$. This framework encompasses sentence translation tasks~\cite{10.5555/2969033.2969173}, speech recognition~\cite{Prabhavalkar2017}, image captioning~\cite{Vinyals_2015_CVPR}, and many more practical problems.

Sequence modeling today is typically performed using neural network based generative models, or \emph{neural sequence models}. Generally, these models are parameterized functions that take as input the sequence $\bm{x}$ and output a sample from the conditional distribution $q\left(\bm{y}\mid\bm{x}\right)$. The parameters of these functions are trained to minimize an appropriate loss function, such as the (forward) empirical cross entropy:
\begin{equation}
        \hat{H}\left(p,q\right)=-\frac{1}{\left\lvert\mathcal{T}\right\rvert}\sum\limits_{\left(\bm{x},\bm{y}\right)\in\mathcal{T}}p\left(\bm{y}\mid\bm{x}\right)\log\left(q\left(\bm{y}\mid\bm{x}\right)\right),
        \label{eq:emp_cross_ent_main_text}
\end{equation}
where $\mathcal{T}=\left\{\left(\bm{x_i},\bm{y_i}\right)\right\}$ are samples from $p\left(\bm{x},\bm{y}\right)$. The backward empirical cross entropy is similarly defined, with $p\leftrightarrow q$. Note that a model with support on an incorrect translation (i.e. $q\neq 0$, $p=0$) yields an infinite backward cross entropy, and a model failing to have support on a correct translation (i.e. $p\neq 0$, $q=0$) yields an infinite forward cross entropy.

To maintain a resource scaling independent of the input sequence length, neural sequence models usually fall into one of two classes: \emph{online sequence models} (also known as \emph{autoregressive sequence models})~\cite{Hopfield2554,10.1162/neco.1997.9.8.1735,cho-etal-2014-learning}, or \emph{encoder-decoder models} (which include state of the art sequence learning architectures, such as Transformers)~\cite{10.5555/2969033.2969173,10.5555/3295222.3295349}. We focus on online models here, and discuss encoder-decoder models in more detail in Appendix~\ref{sec:seq_learning_background}.

In online models, input tokens $\bm{x_i}$ are translated in sequence to output tokens $\bm{y_i}$ via functions $\mathcal{F}_i$. An unobserved internal memory (or \emph{latent space}) $L$ shared between time steps allows the model to represent long-range correlations in the data. A diagram of the general form of online models is given in Fig.~\ref{fig:online_contextual_model}(a). Generally, there are no restrictions on the forms of $\mathcal{F}_i$, though most neural sequence models are composed of simple smooth (or almost everywhere smooth) functions out of training considerations~\cite{Hopfield2554,10.1162/neco.1997.9.8.1735,cho-etal-2014-learning,10.5555/3295222.3295349}. Here, we generalize from the typical smoothness constraints and consider \emph{locally Lipschitz} maps.

Assuming the codomain of $\mathcal{F}_i$ is $\mathbb{R}^m$, all maps that are almost everywhere differentiable with locally bounded Jacobian norm are locally Lipschitz~\cite{Federer1996}. Realistically, then, locally Lipschitz models can be thought of as all models trainable using gradient based methods. Equivalently, they can be thought of as models not infinitely sensitive to infinitesimal changes in their inputs. This includes all models with standard nonlinearities, including those with ReLU, hyperbolic tangent, and sigmoid activation functions. Note that this condition is much weaker than a \emph{globally Lipschitz} constraint. We give a formal definition of local Lipschitzness in Appendix~\ref{sec:seq_learning_background}.

Though neural networks are often described as functions of real-valued inputs, in practice they are implemented at finite precision. We emphasize that where we analytically consider such networks here---such as in Sec.~\ref{sec:cont_sep}---we consider the formal description of neural networks, which assumes infinite precision. Our numerical experiments in Sec.~\ref{sec:numerics}, however, give evidence that our analytic results also hold in the finite precision regime. We discuss this in more detail in Appendix~\ref{sec:exp_consds}.

\subsection{Contextual Recurrent Neural Networks}\label{sec:quantum_contextuality_prelims}

\begin{table}
    \begin{center}
        \def\arraystretch{1.5}
        \ssmall
        \begin{tabular}{c|c|c}
            $X_1\left(\alpha\right)$ & $X_2\left(\alpha\right)$ & $X_1\left(\alpha\right)^\dagger X_2\left(\alpha\right)^\dagger$\\\hline
            $X_1\left(\alpha\right)^\dagger Z_2\left(\frac{\cpi}{2\alpha}\right)^\dagger$ & $Z_1\left(\frac{\cpi}{2\alpha}\right)^\dagger X_2\left(\alpha\right)^\dagger$ & $-X_1\left(\alpha\right)Z_1\left(\frac{\cpi}{2\alpha}\right)X_2\left(\alpha\right)Z_2\left(\frac{\cpi}{2\alpha}\right)$\\\hline
            $Z_2\left(\frac{\cpi}{2\alpha}\right)$ & $Z_1\left(\frac{\cpi}{2\alpha}\right)$ & $Z_1\left(\frac{\cpi}{2\alpha}\right)^\dagger Z_2\left(\frac{\cpi}{2\alpha}\right)^\dagger$
        \end{tabular}
        \caption{An example of CV quantum contextuality using a Mermin--Peres magic square~\cite{PhysRevLett.65.3373}, with CV Pauli operators $X_i\left(a\right),Z_i\left(a\right)$ generated by $-2\ci a\hat{p}_i,2\ci a\hat{q}_i$, respectively. For any real $\alpha\neq 0$, all operators in each row and column commute. Additionally, the product of each row and column is the identity operator, except for the final column, which gives $-1$. Thus, definite classical values cannot be assigned to each operator without yielding a contradiction.\label{table:cv_stab_example}}
    \end{center}
\end{table}
We now consider a quantization of a simple online model. Generally, online models can be interpreted as a classical dynamical process, where queries $\bm{x_i}$ are made to a physical system described by the latent state $\bm{\lambda_{i-1}}$, yielding a result $\bm{y_i}$ and transforming the latent state $\bm{\lambda_{i-1}}\mapsto\bm{\lambda_i}$ (see Fig.~\ref{fig:online_contextual_model}(a)). For \emph{linear recurrent neural networks} (LRNNs), this can be interpreted as the physical process of querying properties of an underlying system described by $\bm{\lambda_i}$ undergoing Hamiltonian evolution under a quadratic Hamiltonian; this can be seen straightforwardly from Hamilton's equations and the linearity of the model. When quantizing the canonical position and momentum variables to operators satisfying the canonical commutation relations, such a model can then be interpreted as performing sequential measurements on a system undergoing evolution via \emph{Gaussian operations}. When these measurements are restricted to homodyne measurements and all inputs are Gaussian states, this process can be simulated classically with memory linear in the number of modes of the Gaussian system~\cite{PhysRevLett.109.230503}. We minimally extend this, and allow for \emph{non-Gaussian measurements}. In particular, we are here interested in measuring via phase estimation the CV analogues of the Pauli operators~\cite{RevModPhys.77.513} (in units where $\hbar=\frac{1}{2}$):
\begin{equation}
    X_i\left(a\right)=\ce^{-2\ci a\hat{p}_i},\hspace{1cm}Z_i\left(a\right)=\ce^{2\ci a\hat{q}_i}.
    \label{eq:cv_x_and_z}
\end{equation}
We also promote the initial state of the network to a \emph{GKP state}~\cite{PhysRevA.64.012310}, which is an eigenstate of CV Pauli operators. We call a recurrent online model beginning in a GKP state, with cell that takes as input $\bm{x_i}$ a description of a CV Pauli operator and returns its measurement result $\bm{y_i}$, a \emph{contextual recurrent neural network} (CRNN). 

This measurement can formally be performed at infinite precision using Gaussian operations and homodyne measurement with fixed ancilla GKP states~\cite{PhysRevA.64.012310,PhysRevA.93.052304}. A circuit description of this is given in Fig.~\ref{fig:online_contextual_model}(b), where $\ket{a}$ is a uniform superposition over squeezed states $\ket{s}$ with $\hat{q}\ket{s}=q\ket{s}$, where $q\equiv 0\pmod{2\cpi}$. When performed sequentially on an initial GKP state, these measurements are what we consider when we compare in Sec.~\ref{sec:cont_sep} CRNNs against the infinite precision classical neural networks described in Sec.~\ref{sec:classical_sequence_learning}. In this scenario, the model is not universal for CV quantum computation, even when additional Gaussian operations within the latent space are added~\cite{calcluth2022}. Counterintuitively, when the initial state is the vacuum state or a finitely squeezed GKP state, the model \emph{is} universal~\cite{PhysRevLett.123.200502,calcluth2022}; this suggests a potential superpolynomial advantage in the expressive power and the time complexity of inference when implemented at finite precision. We discuss this in more detail in Appendix~\ref{sec:exp_consds}.

Just as in the classical case, one can consider a finite precision approximation of these measurements. In this scenario, phase estimation using ancilla qubits can be performed for each measurement~\cite{nielsen_chuang_2010}. We discuss proposals for the experimental implementation of such a measurement in Appendix~\ref{sec:exp_consds}. In general, parameterized Gaussian operations can be included within each recurrent cell to yield a \emph{trainable CRNN}. This is a special case of the CV neural networks considered in \cite{PhysRevResearch.1.033063}, which also considered the training of such networks. For our expressivity separations, however, we consider the fixed CRNN instance given in Fig.~\ref{fig:online_contextual_model}(b).

For our purposes, these measurements are important as CV Pauli operators exhibit \emph{quantum contextuality}~\cite{PhysRevA.82.022114}, in complete analogy with the contextuality present in qubit Pauli operators~\cite{PhysRevLett.65.3373}. Quantum contextuality is the statement that no definite classical values can be assigned to quantum operators, even when the measured operators in any given measurement scenario commute. For an example of this phenomenon, see Table~\ref{table:cv_stab_example}; there is no consistent assignment of classical values to each operator in the Table for any real $\alpha\neq 0$.

\subsection{Stabilizer Measurement Translation}\label{sec:stab_meas_task}

We now focus on a classical sequence learning task that is naturally performed by the introduced CRNN. In particular, we consider the \emph{$\left(k,n\right)$ stabilizer measurement translation task}, parameterized by $k$ and $n$. Leaving the formal definition for Appendix~\ref{sec:proof_of_express_sep}, we give an informal definition here. We use the terminology of Fig.~\ref{fig:online_contextual_model} for clarity.
\begin{definition}[$\left(k,n\right)$ stabilizer measurement translation task, informal]
    Given a $k$ long sequence of classical descriptions $\bm{x_i}$ of CV Pauli operators on $n$ modes, output a sequence of measurement outcomes $y_i$ that is consistent with measuring these operators sequentially on a fixed GKP state $\ket{\lambda_0}$.
    \label{def:stab_trans_task_inf}
\end{definition}
As described in Sec.~\ref{sec:quantum_contextuality_prelims}, such measurement sequences can display nontrivial correlations due to quantum contextuality. Note that this task is distinct from the measurement of position and momentum operators. Here, we require the measurement of linear combinations of position and momentum operators \emph{modulo $2\cpi$}, as we are measuring the phases of operators generated by position and momentum. This can be done using the CRNN cell described in Sec.~\ref{sec:quantum_contextuality_prelims}. We consider in Appendix~\ref{sec:proof_of_express_sep} a slight generalization of this task, though here we consider Definition~\ref{def:stab_trans_task_inf} with its fixed GKP initial state for simplicity.

\section{Bounds on Stabilizer Measurement Translation}\label{sec:cont_sep}

We now give statements and proof sketches of our main results, which are lower bounds on the performance of classical models in performing the stabilizer measurement translation task described in Sec.~\ref{sec:stab_meas_task}. This will give an expressivity separation between the classical and quantum sequence models.

\begin{figure}
    \begin{center}
        \includegraphics[width=0.85\linewidth]{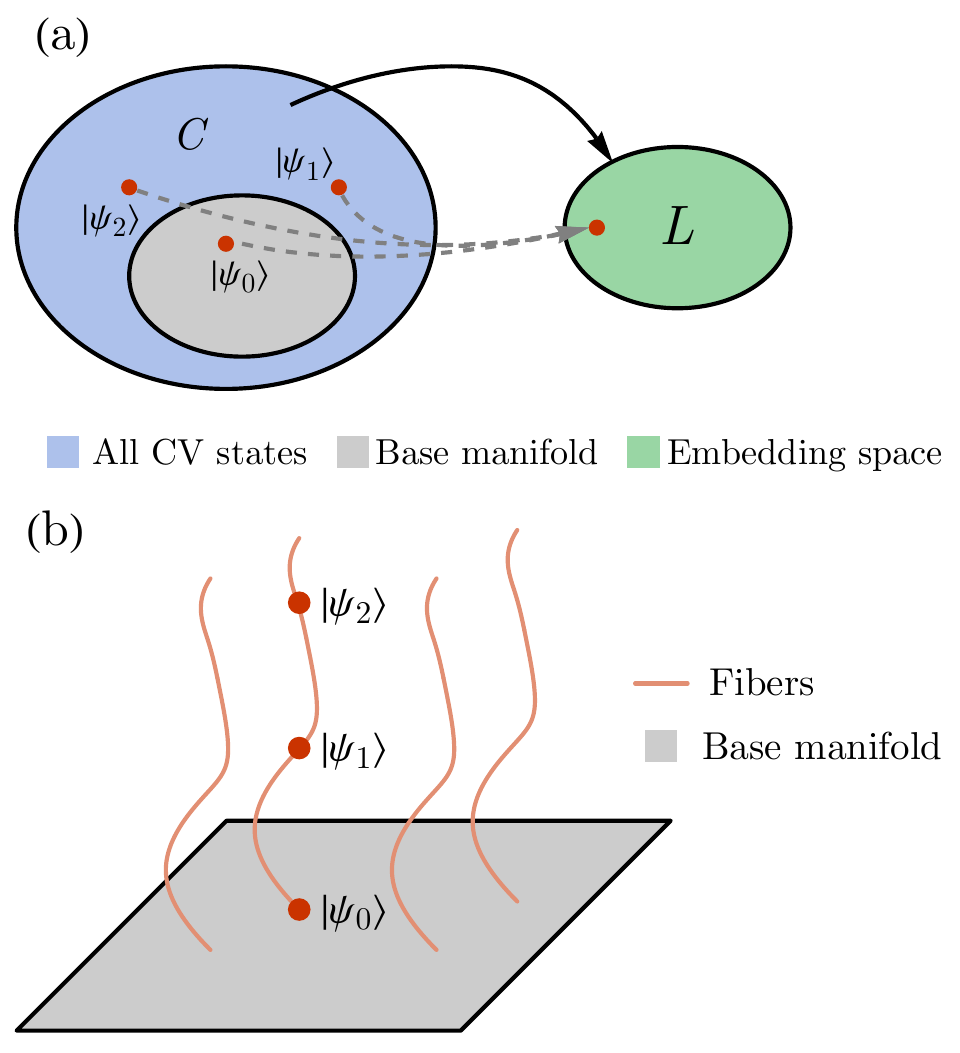}
        \caption{(a) A schematic of the classical model when the dimension of the latent space $L$ (green) is less than $\frac{n\left(n-3\right)}{2}$, where $n$ is the number of modes in the stabilizer measurement translation task. We show that when this is the case, in the neighborhood of some input, only a subspace of inputs (gray) of the same dimension as $L$ are mapped injectively. (b) A sketch of the space of inputs, with fibers locally induced by the model. The base manifold is mapped injectively to $L$. All points on a fiber (e.g. $\ket{\psi_1}$, $\ket{\psi_2}$) map to the same point as their base point (e.g. $\ket{\psi_0}$) in $L$. When the dimension of the fiber is large enough, we show that these states have contextual stabilizers. We then show that this implies that the states have a single-shot distinguishing measurement sequence.\label{fig:proof_schematic}}
    \end{center}
\end{figure}
For discrete models, quantum contextuality was the key resource for showing a separation in expressivity between classical and quantum models~\cite{gao2021enhancing}. Using different proof techniques, we here show that quantum contextuality is also the resource giving the separation between continuous classical and quantum models with infinite dimensional Hilbert spaces. To do this, we specialize to two classes of models: online neural sequence models, and encoder-decoder models (which include state of the art models such as seq2seq models~\cite{10.5555/2969033.2969173} and Transformers~\cite{10.5555/3295222.3295349}). Here, we focus on the memory separation between CRNNs and classical online neural sequence models, and discuss a similar separation against encoder-decoder models in Appendix~\ref{sec:proof_of_express_sep}. We also there formulate a general statement on the classical efficiency of simulating CV Pauli measurements on an initial GKP state, similar in spirit to the fact that Gottesmann--Knill~\cite{gottesman1997,PhysRevA.70.052328} is optimal for qubit stabilizer simulation~\cite{karanjai2018contextuality}.

Our main result can be informally stated as the following Theorem (with the full statement and proof left to Appendix~\ref{sec:proof_of_express_sep}). Note as discussed in Sec.~\ref{sec:stab_meas_task} that a CRNN can perform the stabilizer measurement translation task with $n$ qumodes of memory.
\begin{theorem}[Online stabilizer measurement translation memory lower bound, informal]
    Consider a locally Lipschitz online model with latent space $L$. If $\dim\left(L\right)<\frac{n\left(n-3\right)}{2}$, this model cannot achieve a finite backward cross entropy on the $\left(n+2,n\right)$ stabilizer measurement translation task.
    \label{thm:online_sep_inf}
\end{theorem}
\begin{proofs}
    The strategy of our proof is to show that, when the dimension of $L$ is less than $\frac{n\left(n-3\right)}{2}$, the model must map an embedded submanifold $K$ of the space of the first $n$ inputs to the same point in $L$; in other words, the model loses the ability to distinguish between inputs in $K$. The nontrivial aspect of this proof is to demonstrate that such a $K$ exists, where distinct points in $K$ yield different translations. As the model is unable to distinguish between points in $K$, this then demonstrates that the model will get a translation incorrect, corresponding to an infinite backward cross entropy on the stabilizer measurement translation task. This is equivalent to demonstrating that the quantum mechanical processes being described by points in $K$ yield quantum states that are single-shot distinguishable.
    
    To demonstrate that such a $K$ exists, we use the local Lipschitzness of the model and the constant rank theorem~\cite{butler_timourian_viger_1988}. This then implies that the map $\mathcal{F}$---given by the $n$-fold composition of the $\mathcal{F}_i$ as shown in Fig.~\ref{fig:online_contextual_model}(a)---locally induces a fiber bundle on the input space (as shown in Fig.~\ref{fig:proof_schematic}), where $\mathcal{F}$ can be considered a projection onto the base manifold of this induced fiber bundle. We will slightly abuse notation in the remainder of this proof sketch, and conflate the first $n$ inputs (and their associated outputs) with the quantum state that arises from this measurement sequence.
    
    We consider a fiber of this fiber bundle, with the goal of proving that there exist points in this fiber that are single-shot distinguishable. We show in Appendix~\ref{sec:proof_of_express_sep} that the dimension of this fiber is large enough such that there exist three states in this fiber with stabilizers which exhibit quantum contextuality. We claim (and prove in Appendix~\ref{sec:proof_of_express_sep}) that due to the presence of quantum contextuality in these stabilizers, these states have a distinguishing measurement sequence of length two. When performing this distinguishing measurement sequence, then, the model must give the incorrect measurement results for one of the three states, giving the lower bound on classical simulation. It is easy to see from Eq.~\eqref{eq:emp_cross_ent_main_text} that this yields both an infinite backward cross entropy when these sequences are in the data set $\mathcal{T}$.
    
    \begin{figure}
        \begin{center}
            \includegraphics[width=0.9\linewidth]{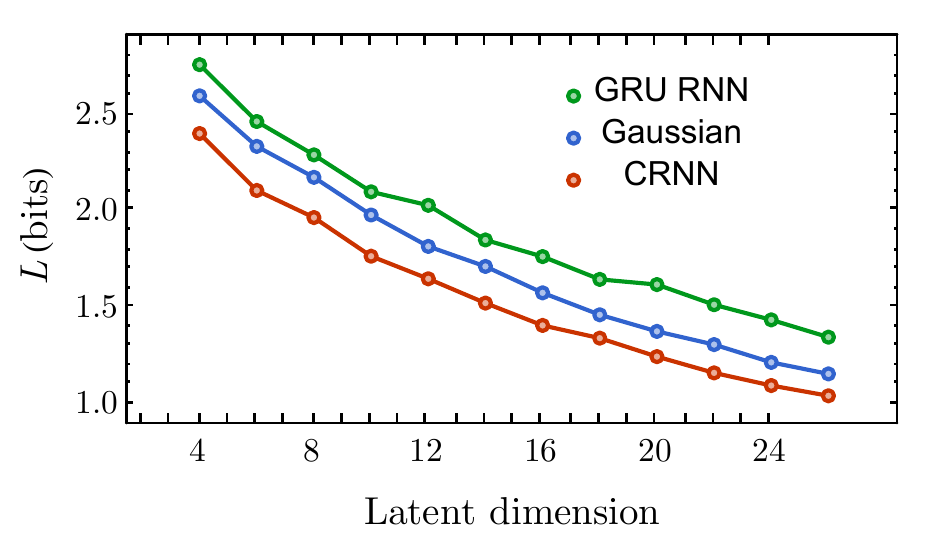}
            \caption{The forward empirical cross entropy ($L$) on a test set for a Spanish-to-English translation task as a function of the model dimension $n$ for GRU RNNs, Gaussian RNNs, and CRNNs. The models are constrained such that the Gaussian and CRNN models have an identical number of parameters. The recurrent cells of the GRU RNN and quantum models have a number of parameters within 2.5\% of each other at $n=26$.\label{fig:cross_entropy_same_params}}
        \end{center}
    \end{figure}
    \begin{table}
        \small
        \begin{center}
            \begin{tabular}{cc}
            \hline
            \textbf{Input}& ``Debemos limpiar la cocina.''\\
            \textbf{Truth}& ``We must clean up the kitchen.''\\
            \textbf{CRNN}& ``We must clean the kitchen.''\\
            \textbf{GRU}& ``We have to turn the right address.''\\
            \hline\hline
            \textbf{Input}& ``Admití que estaba equivocada.''\\ 
            \textbf{Truth}& ``I admitted that I was wrong.''\\
            \textbf{CRNN}& ``I was wrong to say that.''\\
            \textbf{GRU}& ``They had a thing to be true.''\\
            \hline \hline
            \textbf{Input}& ``¿Cual es el lugar m\'{a}s bonito del mundo?''\\
            \textbf{Truth}& ``What’s the most beautiful place in the world?''\\
            \textbf{CRNN}& ``What’s the world largest place?''\\
            \textbf{GRU}& ``What’s the best of is in?''
            \\
            \hline\hline
            \textbf{Input}& ``La caja es pesada.''\\
            \textbf{Truth}& ``The box is heavy.''\\
            \textbf{CRNN}& ``The box is heavy.''\\
            \textbf{GRU}& ``My box is.''
            \\
            \hline
            \end{tabular}
        \caption{Random samples of translation results for $n=26$ models.\label{tab:samples}}
        \end{center}
    \end{table}
    As a simple example of this phenomenon, assume that three states $\ket{\psi_1},\ket{\psi_2},\ket{\psi_3}$ with classical representations in the same fiber are stabilized respectively by the rows of Table~\ref{table:cv_stab_example} for some real $\alpha\neq 0$. As $\ket{\psi_3}$ is stabilized by $Z_1\left(\frac{\cpi}{2\alpha}\right)^\dagger Z_2\left(\frac{\cpi}{2\alpha}\right)^\dagger$, upon measuring this operator, the measurement result is constrained to be $1$ for the simulation to be accurate. The post-measurement state of $\ket{\psi_1}$ is then stabilized by $X_1\left(\alpha\right)X_2\left(\alpha\right)$ and $Z_1\left(\frac{\cpi}{2\alpha}\right)^\dagger Z_2\left(\frac{\cpi}{2\alpha}\right)^\dagger$. In particular, it is also stabilized by $X_1\left(\alpha\right)Z_1\left(\frac{\cpi}{2\alpha}\right)^\dagger X_2\left(\alpha\right)Z_2\left(\frac{\cpi}{2\alpha}\right)^\dagger$. Conversely, the post-measurement state of $\ket{\psi_2}$ is stabilized by $-X_1\left(\alpha\right)Z_1\left(\frac{\cpi}{2\alpha}\right)^\dagger X_2\left(\alpha\right)Z_2\left(\frac{\cpi}{2\alpha}\right)^\dagger$. Thus, then measuring $X_1\left(\alpha\right)Z_1\left(\frac{\cpi}{2\alpha}\right)^\dagger X_2\left(\alpha\right)Z_2\left(\frac{\cpi}{2\alpha}\right)^\dagger$ gives an incorrect translation for one of these states. This measurement sequence is single-shot, as only a single copy of the state being measured is used.
\end{proofs}

Our results show that there is a general $n$ versus $\operatorname{\Omega}\left(n^2\right)$ bound in the memory requirements of contextual and classical models performing the stabilizer measurement translation task. In practice, this can yield an even greater separation in time complexity for given implementations of these models, as the time complexity of inference using classical models is typically superlinear in the model size. We discuss this in more detail in Appendix~\ref{sec:time_complexity}.

\section{Numerical Experiments}\label{sec:numerics}

We now showcase the practical benefit of finding an interpretable advantage in the expressivity of our quantum model. Namely, it is able to give us intuition as to which data sets---beyond the constructed data set used in our proof---a CRNN may outperform classical machine learning models on. As previously discussed, the contextuality present in quantum operators behaves qualitatively similar to the linguistic contextuality present in language. That is, words can have one of many meanings, and their exact definition only becomes apparent when considering their context in a sequence. This is important for translation tasks, where different meanings of a single word in one language have different translations in other languages.

To explore this intuition, we consider the application of a CRNN on a standard Spanish-to-English translation data set~\cite{spaengdata}, with trainable Gaussian interactions within each recurrent cell. We also consider the performance of GRUs~\cite{cho-etal-2014-learning} in a seq2seq learning framework~\cite{10.5555/2969033.2969173}, and Gaussian models (with Gaussian measurements). Details of our numerical simulations for all of the models we consider are given in Appendix~\ref{sec:numerics_deets}, along with details of the architectures. We also discuss in Appendices~\ref{sec:sim_deets} and~\ref{sec:numerics_deets} a $\operatorname{\Theta}\left(n^2\right)$ memory classical simulation of CRNNs with $n$ latent modes on a restricted space of Gaussian operations, which is what we use in our numerical simulations. The Gaussian and contextual models were constrained to have exactly the same number of trainable parameters, and each recurrent cell of the GRU had a parameter count within $2.5\%$ of those of the quantum models at the largest model size considered. In Fig.~\ref{fig:cross_entropy_same_params}, we plot the final training performance of all of our models. It is easy to see that the contextual model outperforms all models under consideration in forward empirical cross entropy at a wide range of model dimensions $n$. Random samples of translation results after training are shown in Table~\ref{tab:samples}.

We also compared the performance of CRNNs against classical linear RNNs (LRNNs)~\cite{pmlr-v70-jing17a} and Transformers~\cite{10.5555/3295222.3295349}. We found that CRNNs substantially outperformed LRNNs. Though the memory requirements of Transformers grow with the length of the input---making direct comparisons against CRNNs difficult---we found that roughly the performance of the largest CRNNs we considered was matched by Transformers with quadratically more memory. We give details of these results in Appendix~\ref{sec:sup_nums}.

\section{Outlook}\label{sec:outlook}

Our results pinpoint quantum contextuality as a resource that can be used to enhance traditional machine learning models. We achieved this by constructing a sequence learning task parameterized by $n$ that a contextual quantum model (a CRNN) of size $n$ is able to model, yet provably no classical neural networks of size $\operatorname{o}\left(n^2\right)$ can model due to their noncontextuality. To our knowledge, this is the first unconditional proof of an expressivity separation between a quantum neural network and classical neural networks on classical data. By explicitly demonstrating that quantum contextuality is the source of this advantage, we are also able to provide intuition as to which classes of problems CRNNs are able to outperform traditional machine learning models in solving. Our numerics confirm the intuition that CRNNs perform extremely well on problems exhibiting linguistic contextuality, such as the Spanish-to-English translation task we consider here.

The simple structure of CRNNs also allow (finite precision approximations of) them to be more amenable to potential experimental implementations when compared with completely general quantum architectures. In particular, all operations in the contextual model are Gaussian, up to the requirement for interactions with fixed ancilla states to perform the required non-Gaussian measurements. The restricted nature of the model may also circumvent the poor training landscapes of generic quantum neural networks~\cite{mcclean2018barren,cerezo2021cost,marrero2021entanglement,napp2022quantifying,anschuetz2021critical,you2021exponentially,anschuetz2022barren}, though we leave further investigation of this to future work.

We believe that the specifics of the CRNN architecture can be relaxed somewhat. Due to recent results linking non-Gaussian operations to quantum contextuality~\cite{booth2021,haferkamp2021}, we suspect any non-Gaussian measurement would make a suitable replacement for the stabilizer measurements we consider here for technical reasons. We also suspect that the technical requirement that the measurements be made with infinite precision to be an artifact of the nature of our proof, which compares the quantum architecture with infinite precision classical models. We believe that in practice, performing phase estimation using ancilla qubits instead of GKP (or other non-Gaussian CV) states is all that is necessary for a practical separation. In fact, such a finite precision implementation may counterintuitively yield a \emph{larger} quantum advantage, as our architecture implemented with a finitely squeezed initial Gaussian state is universal for CV quantum computation~\cite{PhysRevLett.123.200502,calcluth2022}. We discuss these two points in more detail in Appendix~\ref{sec:exp_consds}.

CRNNs demonstrate that even the quantization of a very simple class of classical architectures---here, the class of LRNNs---is able to outperform a wide range of classical models on certain tasks, even if the classical models are much more powerful than LRNNs. We leave for future work the quantization of more powerful classical architectures, which may achieve a practical quantum advantage on a wider variety of tasks than we consider here.

\begin{acknowledgments}
    We thank Li Jing for sharing with us the code used in \cite{pmlr-v70-jing17a}. We also thank Honghao Fu, Pranav Gokhale, Liang Jiang, Robert Huang, Bobak T. Kiani, Seth Lloyd, Mikhail D. Lukin, and Quntao Zhuang for insightful discussions. E.R.A.\ is supported by the National Science Foundation Graduate Research Fellowship Program under Grant No.\ 4000063445, and a Lester Wolfe Fellowship and the Henry W. Kendall Fellowship Fund from M.I.T. H.Y.H.\ is supported by the UC Hellman
    Fellowship and the Harvard Quantum Initiative Fellowship. J.L.H.\ is supported by the Simons Collaboration on Ultra-Quantum Matter, which is a grant from the Simons Foundation (652264). X.G.\ is supported by the Postdoctoral Fellowship in Quantum Science of the MPHQ, the Templeton Religion Trust Grant No.\ TRT 0159, and by the Army Research Office under Grant No.\ W911NF1910302 and MURI Grant No.\ W911NF2010082.
\end{acknowledgments}

\bibliography{main}

\appendix

\onecolumngrid

\begin{spacing}{1.5}
    \section{Background on Sequence Learning}\label{sec:seq_learning_background}
    
    \subsection{Sequence Learning}

    \emph{Sequence-to-sequence} or \emph{sequence} learning~\cite{10.5555/2969033.2969173} is the approximation of some given conditional distribution $p\left(\bm{y}\mid\bm{x}\right)$ with a model distribution $q\left(\bm{y}\mid\bm{x}\right)$. This framework encompasses sentence translation tasks~\cite{10.5555/2969033.2969173}, speech recognition~\cite{Prabhavalkar2017}, image captioning~\cite{Vinyals_2015_CVPR}, and many more practical problems.
    
    The training and evaluation of sequence-to-sequence models is most often performed on the \emph{forward (conditional) cross entropy}:
    \begin{equation}
        H\left(p,q\right)=-\int\dd{\bm{x}}\int\dd{\bm{y}}p\left(\bm{x},\bm{y}\right)\log\left(q\left(\bm{y}\mid\bm{x}\right)\right).
    \end{equation}
    Here, ``forward'' indicates the ordering of the arguments of $H$; the \emph{backward cross entropy} is given by $p\leftrightarrow q$. Given a finite test set $\mathcal{T}=\left\{\left(\bm{x_i},\bm{y_i}\right)\right\}$ of $M$ points sampled from $p\left(\bm{x},\bm{y}\right)$, we can also define the forward \emph{empirical} cross entropy
    \begin{equation}
        \hat{H}\left(p,q\right)=-\frac{1}{M}\sum\limits_{\left(\bm{x},\bm{y}\right)\in\mathcal{T}}p\left(\bm{y}\mid\bm{x}\right)\log\left(q\left(\bm{y}\mid\bm{x}\right)\right),
        \label{eq:emp_cross_ent}
    \end{equation}
    with the backward empirical cross entropy once again given by $p\leftrightarrow q$.
    
    Historically, sequence learning was performed using Bayesian networks such as hidden Markov models~\cite{pearl1985bayesian,10.1214/aoms/1177699147}. However, in recent years, the performance of these models have been eclipsed by neural network based models.
    
    \subsection{Neural Sequence Models}
    
    \begin{figure}
        \begin{center}
            \includegraphics[width=0.4\linewidth]{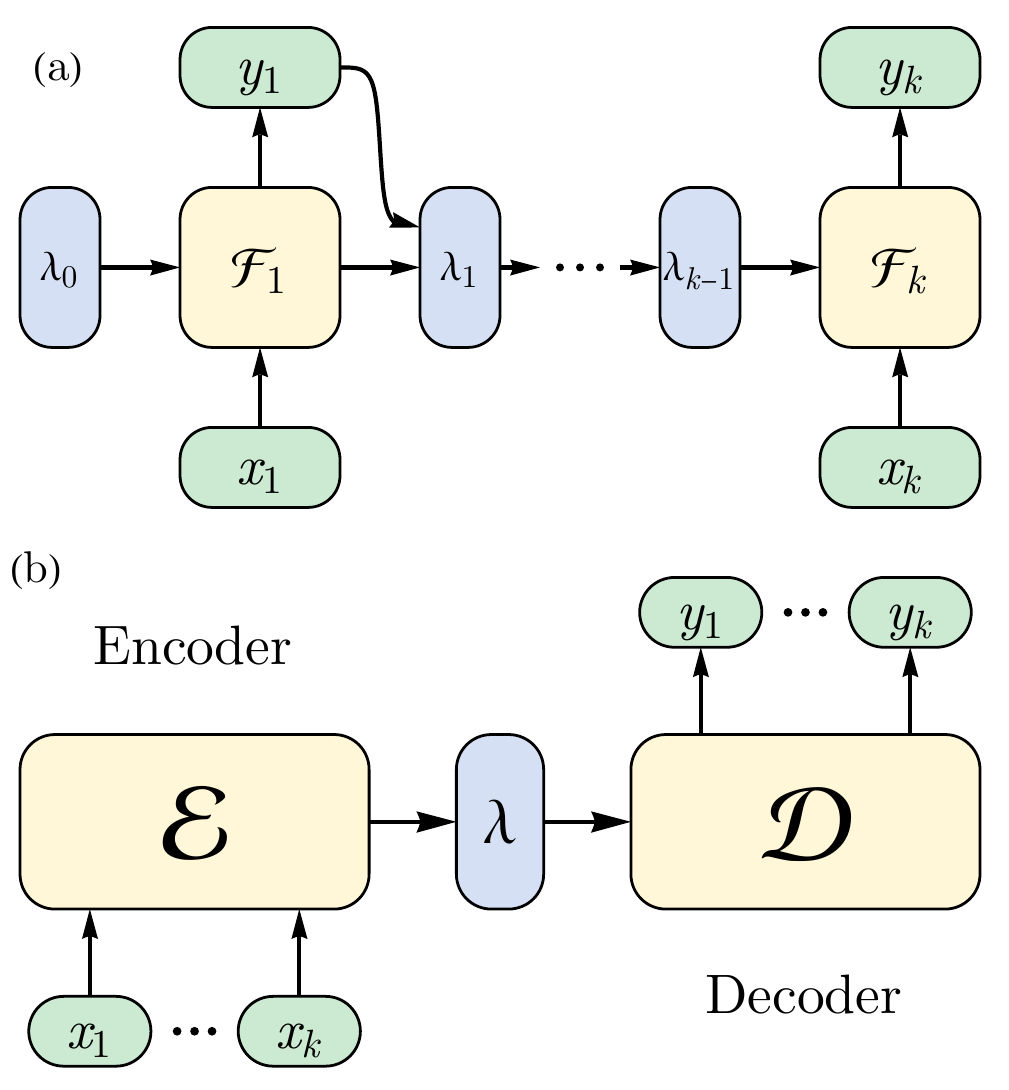}
            \caption{(a) An online neural sequence model. The model autoregressively takes input tokens $\bm{x_i}$, and outputs decoded tokens $\bm{y_i}$, with map $\mathcal{F}_i$. The model also has an unobserved internal memory with state $\bm{\lambda_i}\in L$ after decoding token $i$ that $\mathcal{F}_{i+1}$ can depend on. (b) A general encoder-decoder model. $\mathcal{E}$ encodes the input $\bm{x}$ into some latent representation $\bm{\lambda}\in L$. A decoder $\mathcal{D}$ then outputs the decoded sequence $\bm{y}$.\label{fig:encoder_decoder_model}}
        \end{center}
    \end{figure}
    Sequence modeling today is typically performed using neural network based generative models, or \emph{neural sequence models}. Generally, these models are parameterized functions that take as input the sequence $\bm{x}$ and output a sample from the conditional distribution $p\left(\bm{y}\mid\bm{x}\right)$; the parameters of these functions are trained to minimize an appropriate loss function, such as the empirical cross entropy of Eq.~\eqref{eq:emp_cross_ent}.
    
    To maintain a resource scaling independent of the input sequence length, neural sequence models usually are one of two classes: \emph{online sequence models} (also known as \emph{autoregressive sequence models})~\cite{Hopfield2554,10.1162/neco.1997.9.8.1735,cho-etal-2014-learning}, or \emph{encoder-decoder models}~\cite{10.5555/2969033.2969173,10.5555/3295222.3295349}. Examples of both are given in Fig.~\ref{fig:encoder_decoder_model}. In the former class of models, input tokens $\bm{x_i}$ are translated in sequence to output tokens $\bm{y_i}$ via functions $\mathcal{F}_i$. An unobserved internal memory (or \emph{latent space}) shared between time steps allows the model to represent long-range correlations in the data.
    
    In the latter of these models, an encoder $\mathcal{E}$ maps the input sequence $\bm{x}$ to a latent space representation $\bm{\lambda}\in L$; then, a decoder $\mathcal{D}$ transforms this representation to the output sentence $\bm{y}$. The advantage of encoder-decoder models over generic representations of $p\left(\bm{y}\mid\bm{x}\right)$ is the improved time complexity when considering a lower dimensional representation $\bm{\lambda}$ of $\bm{x}$. When the encoder map is trivial (i.e. when $L$ is congruent to the input space and $\mathcal{E}$ is the identity), then no compression occurs, and the model is equivalent to a general representation of $p\left(\bm{y}\mid\bm{x}\right)$ given by $\mathcal{D}$.
    
    Generally, there are no restrictions on the forms of $\mathcal{F}_i$, $\mathcal{E}$, or $\mathcal{D}$, though most neural sequence models are composed of simple smooth (or almost everywhere smooth) functions out of training considerations~\cite{Hopfield2554,10.1162/neco.1997.9.8.1735,cho-etal-2014-learning,10.5555/3295222.3295349}. Here, we generalize from the typical smoothness constraints and consider \emph{locally Lipschitz} maps. A function $\mathcal{F}:K\to L$ for metric spaces $K$ and $L$ is locally Lipschitz if
    \begin{equation}
        d_L\left(\mathcal{F}\left(\bm{x}\right),\mathcal{F}\left(\bm{x'}\right)\right)\leq C_{\bm{x}}d_K\left(\bm{x},\bm{x'}\right),
    \end{equation}
    where $C_{\bm{x}}$ is constant in some neighborhood of $\bm{x}$. Here, $d_S$ is the distance function on the metric space $S$.
    
    All practical neural sequence models are locally Lipschitz. Indeed, assuming $L=\mathbb{R}^m$, all maps that are almost everywhere differentiable with locally bounded Jacobian norm are locally Lipschitz~\cite{Federer1996}. Realistically, then, locally Lipschitz models can be thought of as all models trainable using gradient based methods; equivalently, they can be thought of as models trainable via methods not arbitrarily sensitive to local noise $\bm{x_i}\mapsto\bm{x_i}+\bm{\epsilon}$.

    \section{Proofs of Expressivity Separations}\label{sec:proof_of_express_sep}
    
    Before giving proofs of expressivity separations between our quantum model and classical models, we first give a formal definition of the translation task we will prove a separation on: namely, \emph{$\left(k,n\right)$ stabilizer measurement translation}, parameterized by $k$ and $n$. Note that the technical description of the task described here is a certain limit of the construction presented in the main text. First, we take $k\to k-n$, as our task is only explicitly defined when the sequence length is at least $n$. For instance, the $\left(n+2,n\right)$ stabilizer measurement translation task as presented in the main text will, here, be referred to as the $\left(2,n\right)$ stabilizer measurement translation task. We make this change as the formal definition of this task as presented here is undefined when $k+n<n$. Second, we now set the first $n$ measurements to be infinite precision Gaussian measurements, i.e. measurements of linear combinations of position and momentum operators. These measurements are a limit of the periodic measurements we consider in the main text, with infinitely large periods.
    
    To be more explicit, we consider an input language given by $n+k$-long sequences of linear combinations of position and momentum operators on $n$ modes. Specifically, input sentences are composed of words which are of the form of rows of:
    \begin{equation}
        \bm{x}=\begin{pmatrix}
        s_{1,1}^q & \ldots & s_{1,n}^q & s_{1,1}^p & \ldots & s_{1,n}^p\\
        \vdots & \vdots & \vdots & \vdots & \vdots & \vdots\\
        s_{n+k,1}^q & \ldots & s_{n+k,n}^q & s_{n+k,1}^p & \ldots & s_{n+k,n}^p
        \end{pmatrix}.
    \end{equation}
    The first $n$ rows describe the sequential measurement of each operator
    \begin{equation}
        \hat{s}_i=\sum\limits_{j=1}^n s_{ij}^q\hat{q}_j+\sum\limits_{j=1}^n s_{ij}^p\hat{p}_j
    \end{equation}
    when beginning in some given fixed state $\ket{\psi_0}$ on $n$ modes that is either a GKP state~\cite{PhysRevA.64.012310} or an infinitely squeezed Gaussian state, which maintains the nonuniversality of the model~\cite{calcluth2022}. The final $k$ rows describe the sequential measurement of each operator
    \begin{equation}
        \hat{s}_i=\exp\left(\ci\sum\limits_{j=1}^n s_{ij}^q\hat{q}_j+\ci\sum\limits_{j=1}^n s_{ij}^p\hat{p}_j\right)
    \end{equation}
    via e.g. phase estimation, as shown in Fig.~\ref{fig:online_contextual_model}(b). Note that the measurement of $\hat{s}_i$ is \emph{not} equivalent to the measurement of its generator. For instance, given states such that:
    \begin{equation}
        \hat{q}\ket{\psi_1}=0,\hspace{1cm}\hat{q}\ket{\psi_2}=2\cpi,
    \end{equation}
    one has
    \begin{equation}
        \left(\exp\left(\ci\hat{q}\right)-1\right)\ket{\psi_1}=\left(\exp\left(\ci\hat{q}\right)-1\right)\ket{\psi_2}=0.
    \end{equation}
    
    A translation $\bm{y}$ of $\bm{x}$ is considered correct if it is of the form
    \begin{equation}
        \bm{b}=\begin{pmatrix}
        m_1\\
        \vdots\\
        m_{n+k}
        \end{pmatrix},
    \end{equation}
    where the measurement outcomes $m_i$ are consistent with those of quantum mechanics. To prove our separations, we will consider input sentences that exhibit quantum contextuality.
    
    \subsection{Expressivity Separation for Online Models}\label{sec:online_sep}
    
    We now show that locally Lipschitz online with latent space dimension less than $\frac{n\left(n-3\right)}{2}$ can stabilizer measurement translate. Our general proof strategy is as follows:
    \begin{enumerate}
        \item We first define a potentially random online learning model with locally Lipschitz cell maps $\mathcal{F}_i^{\bm{r}}\left(\bm{s}_i,\bm{\lambda_{i-1}}\right)=\left(m_i,\bm{\lambda_i}\right)$, where we use $m_i$ to indicate the measurement result when measuring the nullifier described by $\bm{s_i}$, and $\bm{r}$ is a random vector shared between all $\mathcal{F}_i$. We assume for any $\bm{r}$ that $\mathcal{F}_i^{\bm{r}}$ is deterministic. Let $\mathcal{F}^{\bm{r}}\left(\bm{s_1},\ldots,\bm{s_n}\right)=\bm{\lambda_n}$ be the $n$-fold composition of $\mathcal{F}_i^{\bm{r}}$ on some fixed initial $\bm{\lambda_0}$ (where any $m_i$ is implicit, as each $m_i$ is fully determined by $\bm{r}$ and $\bm{s_i}$). Note that once $\bm{r}$ is specified, $\mathcal{F}_{\bm{r}}$ is a deterministic function. Due to this, in the following, we take the $\bm{r}$ dependence to be implicit.
        \item We assume the dimension of $\bm{\lambda_i}$ is less than $\frac{n\left(n-3\right)}{2}$. We prove that then, the described online model will give a wrong measurement outcome on the final two measurement results $m_{n+1},m_{n+2}$. In the following, we will refer to this as \emph{the theorem statement}.
        \item To prove that \emph{the theorem statement} is true, we show that it is true for a subspace $K$ of inputs which describe \emph{CV graph states}, where the associated graphs have no self-loops. It is easy to see that the dimension of such a space is $\frac{n\left(n-1\right)}{2}$, by studying the space of adjacency matrices. We let $\left.\mathcal{F}\right|_K$ be $\mathcal{F}$ restricted to this space $K$. Let $\bm{B}$ be coordinates of $K$, as defined in Eq.~\eqref{eq:q_def}. We assume the Jacobian of this map achieve its maximal rank at $\bm{B}=\bm{0}$, which describes a set of measurements yielding the position squeezed state $\ket{\bm{0}}_{\bm{\hat{q}}}$. By doing so, we guarantee the robustness of the Jacobian rank in a neighborhood of $\bm{B}=\bm{0}$. The assumption that $\bm{B}=\bm{0}$ is a point of maximal rank is taken WLOG, as there exist Gaussian operations that transform the $\bm{B}$ at which the Jacobian of $\left.\mathcal{F}\right|_K$ attains its maximal rank to $\bm{B}=\bm{0}$.
        \item As the rank is constant in the neighborhood of $\bm{B}=\bm{0}$, by the constant rank theorem~\cite{butler_timourian_viger_1988}, $\left.\mathcal{F}_{\bm{r}}\right|_K$ induces a fiber bundle structure in the neighborhood of $\bm{B}=\bm{0}$. That is, $\left.\mathcal{F}_{\bm{r}}\right|_K$ is a projection of fibers in a neighborhood of $\bm{B}=\bm{0}$ to their base points. This means that the model is unable to distinguish between points that share a fiber in this neighborhood.
        \item We then show that when $\dim\left(\bm{\lambda_n}\right)<\dim\left(K\right)-n$, there exist $\bm{B'},\bm{B''}$ on the fiber with base point $\bm{B}=\bm{0}$ such that $\bm{B},\bm{B'},\bm{B''}$ describe distinct states. We show in Lemma~\ref{lemma:graph_state_context} that these states have stabilizers which share quantum contextuality, yielding distinguishing one-shot measurement sequences. As the model is unable to distinguish between $\bm{B},\bm{B'},\bm{B''}$, this implies that there exist $\bm{s_{n+1}},\bm{s_{n+2}}$ describing this distinguishing measurement sequence, that the model returns the wrong measurement result for at least one of $\bm{B},\bm{B'},\bm{B''}$ with certainty. This proves \emph{the theorem statement}. We also use this general proof strategy when considering encoder-decoder models in Theorem~\ref{thm:enc_dec_sep}, up to some minor details.
    \end{enumerate}
    
    With our proof strategy now clear, we now proceed to prove the details. First, we prove our lemma demonstrating that indeed, the stabilizer operators we consider exhibit quantum contextuality. We also show that this contextuality induces a one-shot distinguishing measurement sequence on the states stabilized by the given operators.
    \begin{lemma}[CV graph state stabilizers exhibit quantum contextuality]
        Consider $\ket{\bm{0}}_{\bm{\hat{q}}}$, the state nullified by all $\hat{q}_i$. Consider two states $\ket{\psi_1}$ and $\ket{\psi_2}$ that are CV graph states (up to arbitrary phases on their stabilizers) with no loops with distinct (modulo $\cpi$) adjacency matrices. There exist operators that stabilize these three states that exhibit quantum contextuality. Furthermore, there exists a distinguishing measurement given by one of the stabilizers of $\ket{\bm{0}}_{\bm{\hat{q}}}$ that maps $\ket{\psi_1}$ and $\ket{\psi_2}$ to orthogonal post-measurement states when the measurement result is $1$; in other words, there exists a distinguishing measurement sequence of length two that distinguishes these three states.
        \label{lemma:graph_state_context}
    \end{lemma}
    \begin{proof}
        As $\ket{\psi_1}$ and $\ket{\psi_2}$ are CV graph states with distinct adjacency matrices (modulo $\cpi$), they must differ (modulo $\cpi$) in the edges $\bm{e_i}$ touching some vertex $i$. That is, $\ket{\psi_1}$ is stabilized by some:
        \begin{equation}
            s_i'=\ce^{2\ci\theta'}X_i\left(1\right)\bm{Z}\left(\bm{e_i}'\right),
        \end{equation}
        and $\ket{\psi_2}$ by some:
        \begin{equation}
            s_i''=\ce^{2\ci\theta''}X_i\left(1\right)\bm{Z}\left(\bm{e_i}''\right),
        \end{equation}
        where $\bm{e_i}'$ and $\bm{e_i}''$ differ (modulo $\cpi$) in some element indexed by $j\neq i$ (as there are no loops in either graph), and where $2\theta',2\theta''$ are phases. Here, $\bm{Z}\left(\cdot\right)$ is defined as the tensor product:
        \begin{equation}
            \bm{Z}\left(\bm{v}\right)=\bigotimes_i Z_i\left(v_i\right).
        \end{equation}
        \begin{figure}
            \begin{center}
                \includegraphics[width=0.3\linewidth]{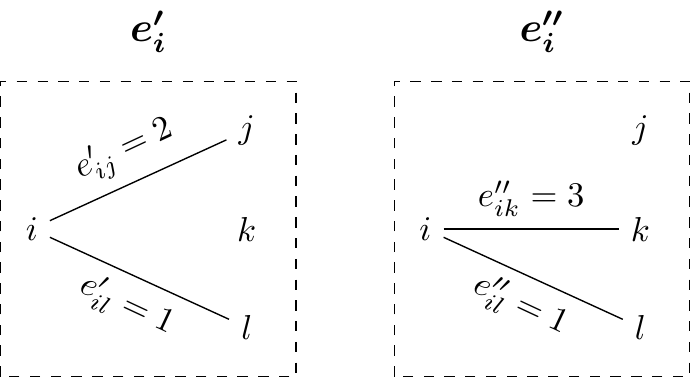}
                \caption{An example of two graphs, with edges leaving vertex $i$ given by $\bm{e_i}'$ and $\bm{e_i}''$, respectively. As the two graphs differ, they must differ by an edge; indeed, they differ in the weights of both $\left\langle i,j\right\rangle$ and $\left\langle i,k\right\rangle$.\label{fig:graph_state_edges}}
            \end{center}
        \end{figure}
        An example diagram of the entries of $\bm{e_i}',\bm{e_i}''$ is given in Fig.~\ref{fig:graph_state_edges}. By the symmetry of CV graph state adjacency matrices, we thus have that the former is also stabilized by
        \begin{equation}
            s_j'=\ce^{2\ci\phi'}X_j\left(1\right)\bm{Z}\left(\bm{e_j}'\right)
        \end{equation}
        and the latter by
        \begin{equation}
            s_j''=\ce^{2\ci\phi''}X_j\left(1\right)\bm{Z}\left(\bm{e_j}''\right),
        \end{equation}
        where $\bm{e_j}'$ and $\bm{e_j}''$ differ (modulo $\cpi$) in some element indexed by $i$, and where $2\phi',2\phi''$ are phases. Note in particular that we have the commutation relations
        \begin{align}
            \left[s_i',s_i''\right]=0,\hspace{1cm}\left[s_j',s_j''\right]=0,\hspace{1cm}&\left[s_i',s_j'\right]=0,\hspace{1cm}\left[s_i'',s_j''\right]=0,\\
            s_i' s_j''=\ce^{2\ci\zeta}s_j''s_i',\hspace{1cm}s_j' s_i''=\ce^{2\ci\zeta}s_i''s_j',\hspace{1cm}&s_i' s_j^{\prime\prime\dagger}=\ce^{-2\ci\zeta}s_j^{\prime\prime\dagger}s_i',\hspace{1cm}s_j' s_i^{\prime\prime\dagger}=\ce^{-2\ci\zeta}s_i^{\prime\prime\dagger}s_j',
        \end{align}
        where
        \begin{equation}
            \zeta\equiv e_{ij}'-e_{ij}''=e_{ji}'-e_{ji}''\neq 0\pmod{\cpi}.
        \end{equation}
        As $\zeta\neq 0\pmod{\cpi}$, there exists some $\alpha\in\mathbb{R}_+^\ast$ such that
        \begin{align}
            \left\{s_i^{\prime\alpha},s_j^{\prime\prime\alpha}\right\}&=0,&\left\{s_j^{\prime\alpha},s_i^{\prime\prime\alpha}\right\}&=0,\\
            \left\{s_i^{\prime\alpha},s_j^{\prime\prime\dagger\alpha}\right\}&=0,&\left\{s_j^{\prime\alpha},s_i^{\prime\prime\dagger\alpha}\right\}&=0.
        \end{align}
        To save on notation, we redefine all stabilizers to be given by their $\alpha$ power, i.e. $s^\alpha\to s$ (and similarly redefine the phases $\theta',\theta'',\phi',\phi''$ by their scaling by $\alpha$). As $\ket{\psi_1},\ket{\psi_2}$ are CV graph states, these rescaled operators are still stabilizers. We also define:
        \begin{align}
            s_i&\equiv\ce^{-2\ci\left(\theta'-\theta''\right)}s_i' s_i^{\prime\prime\dagger},\\
            s_j&\equiv\ce^{-2\ci\left(\phi'-\phi''\right)}s_j' s_j^{\prime\prime\dagger},
        \end{align}
        which are stabilizers of $\ket{\bm{0}}_{\bm{\hat{q}}}$. Thus, we have constructed nine observables with constraints satisfying those of a Mermin--Peres magic square (see Table~\ref{tab:magic_square} for an example), a well-known proof of quantum contextuality~\cite{PhysRevLett.65.3373}.
        \begin{table}
            \begin{center}
                \begin{tabular}{c|c|c}
                    $s_i'$ & $s_j'$ & $s_i^{\prime\dagger}s_j^{\prime\dagger}$\\\hline
                    $s_i^{\prime\prime\dagger}$ & $s_j^{\prime\prime\dagger}$ & $s_i''s_j''$\\\hline
                    $s_i^\dagger$ & $s_j^\dagger$ & $s_i s_j$
                \end{tabular}
                \caption{The Mermin--Peres magic square of stabilizers of states mapping to the same latent space under a locally Lipschitz map (with $\theta'=\theta''=\phi'=\phi''=0$ for simplicity). Stabilizers of $\ket{\psi_1}$, $\ket{\psi_2}$, and $\ket{\bm{0}}_{\bm{\hat{q}}}$ make up the three rows. All observables in each row and column commute. Furthermore, the product of observables in each row and column is the identity, except for the third column, which gives minus the identity. See Table~\ref{table:cv_stab_example} for a special case of this magic square.}
                \label{tab:magic_square}
            \end{center}
        \end{table}
        
        Consider now the post-measurement states  $\ket{\psi_1'},\ket{\psi_2'}$ of $\ket{\psi_1},\ket{\psi_2}$, respectively, when $s_i s_j$ is measured to be $1$. By Table~\ref{tab:magic_square}, $\ket{\psi_1'}$ is stabilized by $s_i' s_j'$ and $s_i s_j$; furthermore, $\ket{\psi_2''}$ is stabilized by $s_i''s_j''$ and $s_i s_j$, and therefore is also stabilized by
        \begin{equation}
            s_i s_j s_i''s_j''=-\ce^{-2\ci\left(\theta'-\theta''+\phi'-\phi''\right)}s_i's_j'.
        \end{equation}
        If
        \begin{equation}
            \theta'-\theta''+\phi'-\phi''\neq\frac{\cpi}{2}\pmod{\cpi},
            \label{eq:phase_rel_dist_meas}
        \end{equation}
        we have that $\ket{\psi_1'}$ and $\ket{\psi_2'}$ are orthogonal. If it is congruent to $\frac{\cpi}{2}\pmod{\cpi}$, then either $\theta'-\theta''\neq 0\pmod{\cpi}$ or $\phi'-\phi''\neq 0\pmod{\cpi}$ (or both). Assume the latter WLOG (the former case is the same with $\phi\to\theta$ and $j\to i$), and instead consider the post-measurement states $\ket{\psi_1'},\ket{\psi_2'}$ of $\ket{\psi_1},\ket{\psi_2}$, respectively, when $s_j$ is measured to be $1$. By Table~\ref{tab:magic_square}, $\ket{\psi_1'}$ is stabilized by $s_j'$ and $s_j$; furthermore, $\ket{\psi_2''}$ is stabilized by $s_j''$ and $s_j$, and therefore is also stabilized by
        \begin{equation}
            s_j s_j''=\ce^{-2\ci\left(\phi'-\phi''\right)}s_j'.
        \end{equation}
        Thus, again, $\ket{\psi_1'}$ and $\ket{\psi_2'}$ are orthogonal.
    \end{proof}
    
    We now consider the locally Lipschitz online learner, with structure given by Fig.~\ref{fig:encoder_decoder_model}(a). We assume that the learner is deterministic, and discuss the extension to randomized models at the end of this Appendix.
    
    Online neural sequence models at time step $i$ map an input token $\bm{x_i}$ and a latent vector $\bm{\lambda_{i-1}}$ to an output token $m_i$ and a new latent vector $\bm{\lambda_i}$. After $n$ steps, then, we can consider the online model as a locally Lipschitz map:
    \begin{equation}
        \mathcal{F}^{\bm{r}}:\left(\mathbb{R}^{2n}\right)^n\to L\times\mathbb{R}^n,
        \label{eq:online_map}
    \end{equation}
    where $L$ is a locally Lipschitz latent manifold and $\bm{r}$ is a random vector such that, for any $\bm{r}$, $\mathcal{F}^{\bm{r}}$ is deterministic. This consideration of $\mathcal{F}^{\bm{r}}$ as a deterministic function of a random $\bm{r}$ is typically the implementation of stochastic learners, such as generative adversarial networks (GANs)~\cite{NIPS2014_5ca3e9b1} and flow-based models~\cite{rezende2015variational}. This also includes implementations of stochastic simulation algorithms such as Wigner function simulation~\cite{PhysRevLett.109.230503}. As for each $\bm{r}$, there exists an input sequence such that any classical model with $\dim\left(L\right)<\frac{n\left(n-3\right)}{2}$ deterministically outputs a measurement sequence inconsistent with quantum mechanics, our results still hold for these classes of random models. Due to this, in the following, we take the $\bm{r}$ dependence to be implicit.
    
    With the preliminaries in place, we now prove our expressivity separation.
    \begin{theorem}[Online stabilizer measurement translation lower bound]
        Consider an online model with locally Lipschitz latent manifold $L$ and locally Lipschitz map $\mathcal{F}$ as described in Eq.~\eqref{eq:online_map}. If $\dim\left(L\right)<\frac{n\left(n-3\right)}{2}$, this model cannot achieve a finite backward empirical cross entropy on the $\left(2,n\right)$ stabilizer measurement translation task.
        \label{thm:online_sep}
    \end{theorem}
    \begin{proof}
        Consider $K\subset\left(\mathbb{R}^{2n}\right)^n$, with elements of the form:
        \begin{equation}
            \bm{Q}=\left(\begin{array}{cc|cc}
            \bm{B}+\bm{H}/2 &&& \left\lVert\bm{B}\right\rVert_{\text{F}}\bm{I_n}
            \end{array}\right);\label{eq:q_def}
        \end{equation}
        here, $\left\lVert\bm{B}\right\rVert_{\text{F}}$ is the Frobenius norm of $\bm{B}$, and $\bm{B}$ is an $n \times n$ hollow (zero diagonal elements) symmetric matrix with entries bounded to be $\left[-\frac{1}{4},\frac{1}{4}\right]$. $\bm{H}$ is the fixed $n \times n$ symmetric hollow matrix of ones, i.e.
        \begin{equation}
            \bm{H}=\begin{pmatrix}
            0 & 1 & \cdots & 1\\
            1 & 0 & \cdots & 1\\
            \vdots & \vdots & \ddots & \vdots\\
            1 & 1 & \cdots & 0\\
            \end{pmatrix}.
        \end{equation}
        It is obvious from this construction that $K$ is an $\frac{n\left(n-1\right)}{2}$-dimensional embedding of a compact subspace of hollow symmetric matrices (with bounded entries and norm) $\bm{B}$. Note that the states described by the measurement scenarios of points in $K$ are exactly CV graph states without loops (with bounded weight edges, as $K$ is compact) and, depending on the measurement results, perhaps overall phases on the stabilizers. To see this, note that the symmetric constraint on $\bm{B}$ ensures that the symplectic product of any two rows of $\bm{Q}$ is zero; furthermore, the final $n$ columns of $\bm{Q}$ are linearly independent for all $\bm{B}\neq\bm{0}$, and the first $n$ columns for $\bm{B}=\bm{0}$ due to the shift by $\bm{H}$. Thus, all points in $K$ are full row rank, and the rows of $\bm{Q}$ completely determine the CV stabilizer state, up to phases given by the measurement results of these operators. Furthermore, as $\bm{B}$ and $\bm{H}$ are hollow, the CV graph state $\bm{Q}$ describes has no loops. Also note that, up to independent rescalings of the rows of $\bm{Q}$, different $\bm{Q}$ correspond to different graph states. We assume WLOG that the Jacobian of $\mathcal{F}$ attains its maximum rank at $\bm{B}=\bm{0}$ (that is, the squeezed state $\ket{\bm{0}}_{\bm{\hat{q}}}$); this can always be done by implicitly transforming the basis of inputs to the model (i.e. by appropriately relabeling points in $\left(\mathbb{R}^{2n}\right)^{n+2}$), and then considering $K$ as previously defined in this new basis.
        
        We will proceed as follows. First, we will show that when $\dim\left(L\right)$ is sufficiently small, $\mathcal{F}$ must map three distinct CV graph states described by different $\bm{Q}$ to the same point in latent space. Then, we will use Lemma~\ref{lemma:graph_state_context} to show that the stabilizers of these states exhibit quantum contextuality (independent of the associated $n$ measurement results), and give rise to a distinguishing measurement sequence. Thus, by considering measurement sequences of length $n+2$ that include these three $\bm{Q}$ and the length two distinguishing measurement sequence, one of the final two measurement outcomes must be incorrect. This implies that there is an infinite backward empirical cross entropy on any finite set containing these three measurement sequences.
        
        Let us begin by showing that $\left.\mathcal{F}\right|_K$ (i.e. the locally Lipschitz map that is $\mathcal{F}$ restricted to $K$) must map three nontrivially distinct $\bm{Q}$ (i.e. three distinct CV graph states) to the same point in latent space when $\dim\left(L\right)$ is sufficiently small. By the constant rank theorem and the local Lipschitzness of $\left.\mathcal{F}\right|_K$, $\left.\mathcal{F}\right|_K$ is not injective for
        \begin{equation}
            \dim\left(L\right)<\dim\left(K\right)=\frac{n\left(n-1\right)}{2}.
        \end{equation}
        In particular, in a sufficiently small neighborhood of $\bm{B}=\bm{0}$ (where the Jacobian of $\left.\mathcal{F}\right|_K$ attains its maximal rank), there exist local coordinates $\bm{\tilde{x}}$ of $K$ and $L$ such that
        \begin{equation}
            \left.\mathcal{F}\right|_K\left(\tilde{x}_1,\ldots,\tilde{x}_{\frac{n\left(n-1\right)}{2}}\right)=\left(\tilde{x}_1,\ldots,\tilde{x}_l,0,\ldots,0\right)
        \end{equation}
        for some $l\leq\dim\left(L\right)<\frac{n\left(n-1\right)}{2}$~\cite{butler_timourian_viger_1988}. WLOG, we identify $\bm{\tilde{x}}=\bm{0}$ with $\bm{B}=\bm{0}$, which is the state infinitely squeezed in all $\hat{q}_i$. We will call $C$ the fiber with local coordinates
        \begin{equation}
            \bm{\tilde{x}}=\left(0,\ldots,0,\tilde{x}_{l+1},\ldots,\tilde{x}_{\frac{n\left(n-1\right)}{2}}\right),
        \end{equation}
        which is of dimension at least
        \begin{equation}
            \varDelta\equiv\frac{n\left(n-1\right)}{2}-\dim\left(L\right)\geq 1.
        \end{equation}
        By construction, all points in $C$---including $\bm{B}=\bm{0}$---map to the same point $l\in L$ under $\left.\mathcal{F}\right|_K$. We now assume that $\dim\left(L\right)<\frac{n\left(n-3\right)}{2}$ such that $\varDelta\geq n+1$.
        
        Now fix $\bm{B}=\bm{0}$ and $\bm{B'}\neq\bm{B}$ in $C$. As described previously, $\bm{Q}$ (and thus $\bm{B}$) completely determines a CV graph state after $n$ measurements, up to independent rescalings of the rows of $\bm{Q}$ (and the measurement results). As the dimension of the space of points that differ (modulo $\cpi$) from $\bm{B'}+\bm{H}/2$ by just a scaling factor in each row is at most $n$, because $\varDelta\geq n+1$ we must have that there exists another $\bm{B''}\neq\bm{B},\bm{B'}$ describing a distinct CV graph state. Therefore, by Lemma~\ref{lemma:graph_state_context}, we have that there exists a distinguishing measurement sequence for these three states. Note that as Lemma~\ref{lemma:graph_state_context} does not depend on the phases of the CV stabilizers, the existence of this distinguishing measurement holds true regardless of what the measurement results are (i.e. independently from what the model outputs for the first $n$ tokens in the decoded sequence). As after $n$ tokens all three sequences map to the same point in latent space in the model, and as they share a distinguishing measurement sequence, the model must obtain an infinite backward empirical cross entropy on these three input sequences when followed by the distinguishing measurement sequence.
    \end{proof}
    
    \subsection{A CV Gottesmann--Knill Lower Bound}

    We now show that our results can be reformulated as a memory lower bound on the classical simulation of stabilizer measurement scenarios. In practice, using finite resources (i.e. at finite precision), any classical ontological model simulating $p\left(\bm{y}\mid\bm{x}\right)$ can only be evaluated at a finite number of $\bm{x}$. We now show that \emph{any} locally Lipschitz interpolation of such a model to real $\bm{x}$ cannot accurately simulate Gaussian operations on an initial GKP state. This includes, for instance, any polynomial interpolation (which always exists).
    \begin{corollary}[CV Gottesmann--Knill lower bound]
        Consider a classical ontological model $p\left(\bm{y}\mid\bm{x}\right)$ with a latent space of dimension less than $\frac{n\left(n-3\right)}{2}$, simulating $\left(k,n\right)$ stabilizer measurement translation with $k\geq 2$. Assume that this ontological model is defined at a finite number of $\bm{x}$. There exists a locally Lipschitz interpolation of this model to all $\bm{x}$. Furthermore, no locally Lipschitz interpolation of this model can faithfully perform stabilizer measurement translation at all $\bm{x}$.
    \end{corollary}
    \begin{proof}
        As there exists a polynomial interpolation of $p$, and as all polynomials of finite degree are locally Lipschitz, there exists a locally Lipschitz interpolation of this model to all $\bm{x}$. Furthermore, no locally Lipschitz interpolation of this model can faithfully perform stabilizer measurement translation by Theorem~\ref{thm:online_sep}, as the composition of locally Lipschitz functions is locally Lipschitz.
    \end{proof}
    
    \subsection{Expressivity Separation for Encoder-Decoder Models}
    
    Though online sequence models are perhaps conceptually the simplest as they directly map input tokens to output tokens, in practice encoder-decoder models outperform them~\cite{10.5555/2969033.2969173,10.5555/3295222.3295349}. We now show that no encoder-decoder model with a locally Lipschitz encoder (and an additional technical assumption) can perform stabilizer measurement translation to finite backward empirical cross entropy. The proof will be similar to that of Theorem~\ref{thm:online_sep}; however, as the model can see the entire input sequence at once, we do not directly have the freedom to choose the distinguishing measurement sequence as in Theorem~\ref{thm:online_sep}. Instead, we will require an input sequence of length quadratic in $n$ (and our additional technical assumption) to force the distinguishing measurement sequence. Note that, as the memory of the contextual learner is independent of the sequence length, this new sequence length has no impact on the memory separation.
    
    We consider an encoder-decoder model with structure given by Fig.~\ref{fig:encoder_decoder_model}(b). The encoder of such a model can be considered a locally Lipschitz map
    \begin{equation}
        \mathcal{E}^{\bm{r}}:\left(\mathbb{R}^{2n}\right)^{n^2}\to L
        \label{eq:model_map}
    \end{equation}
    to some locally Lipschitz latent manifold $L$. As in Appendix~\ref{sec:online_sep}, $\bm{r}$ is a random vector such that, for any $\bm{r}$, $\mathcal{E}^{\bm{r}}$ is deterministic. We once again make the $\bm{r}$ dependence implicit in the following.
    
    For technical reasons, we slightly change the definition of the $\left(k,n\right)$ stabilizer measurement translation task, where now the final $k$ measurement descriptions instead describe the measurements of the operators:
    \begin{equation}
        \hat{s}_i=\exp\left(\ci\sum\limits_{j=1}^n\frac{\bm{1}\left[s_{ij}^q\neq 0\right]}{s_{ij}^q}\hat{q}_j+\ci\sum\limits_{j=1}^n\frac{\bm{1}\left[s_{ij}^p\neq 0\right]}{s_{ij}^p}\hat{p}_j\right),
    \end{equation}
    where we define $\frac{\bm{1}\left[x\neq 0\right]}{x}$ to be zero when $x=0$. We will call this the \emph{modified $\left(k,n\right)$ stabilizer measurement translation task}. This can obviously still be performed perfectly with a CRNN of model size $n$, by either changing the parameters of the phase estimation circuit of Fig.~\ref{fig:online_contextual_model}(b) to be given by $\frac{\bm{1}\left[x_{ij}\neq 0\right]}{x_{ij}}$, or more formally by introducing a quantum circuit computing $\frac{\bm{1}\left[x_{ij}\neq 0\right]}{x_{ij}}$ on which these gates control.
    
    We now discuss our additional technical assumption. Defining the subspace $R$ of inputs as in the proof of Theorem~\ref{thm:enc_dec_sep}, we assume that the Jacobian of the encoder restricted to $R$ attains its maximal rank at some point of the form $\left(\bm{Q},\bm{0}\right)\in R$. A sufficient condition for this is requiring that some point of the form $\left(\bm{Q},\bm{0}\right)$ is not a critical point of $\left.\mathcal{E}\right|_R$; this condition is satisfied by generic $\mathcal{E}$ when $\dim\left(L\right)<\frac{n\left(n-3\right)}{2}$, and also by models with encoders constrained to be submersions. In fact, when the latter holds, then it is easy to see from the proof of Theorem~\ref{thm:enc_dec_sep} that the separation still holds on the \emph{unmodified} $\left(k,n\right)$ stabilizer measurement translation task, as all properties we use that hold locally then hold globally. Any one of these conditions is sufficient, and needed for our proof technique to be able to analyze any neighborhood of the non-Gaussian measurements we consider.
    \begin{theorem}[Encoder-decoder stabilizer measurement translation lower bound]
        Consider an encoder-decoder model with locally Lipschitz latent manifold $L$. Let $\mathcal{E}$ be the associated locally Lipschitz encoder function, as defined in Eq.~\eqref{eq:model_map}, and assume that the Jacobian of the map $\left.\mathcal{E}\right|_R$ (where the subspace $R$ is defined below) attains its maximal rank at some point of the form $\left(\bm{Q},\bm{0}\right)\in R$. If $\dim\left(L\right)<\frac{n\left(n-3\right)}{2}$, this model cannot achieve a finite backward empirical cross entropy on the modified $\left(n^2-n,n\right)$ stabilizer measurement translation task.\label{thm:enc_dec_sep}
    \end{theorem}
    \begin{proof}
        Consider $R\equiv K\times\left(\mathbb{R}^{2n}\right)^{n^2-n}\subset\left(\mathbb{R}^{2n}\right)^n\times\left(\mathbb{R}^{2n}\right)^{n^2-n}\cong\left(\mathbb{R}^{2n}\right)^{n^2}$ with elements $\left(\bm{Q},\bm{P}\right)$ of the following form:
        \begin{enumerate}
            \item $\bm{Q}\in K$ is given by rows of matrices of the form:
            \begin{equation}
                \bm{Q}=\left(\begin{array}{cc|cc}
                \bm{B}+\bm{H}/2 &&& \left\lVert\bm{B}\right\rVert_{\text{F}}\bm{I_n}
                \end{array}\right);
            \end{equation}
            here, $\left\lVert\bm{B}\right\rVert_{\text{F}}$ is the Frobenius norm of $\bm{B}$, and $\bm{B}$ is a hollow symmetric matrix with entries bounded to be $\left[-\frac{1}{4},\frac{1}{4}\right]$. $\bm{H}$ is the fixed symmetric hollow matrix of ones, i.e.
            \begin{equation}
                \bm{H}=\begin{pmatrix}
                0 & 1 & \cdots & 1\\
                1 & 0 & \cdots & 1\\
                \vdots & \vdots & \ddots & \vdots\\
                1 & 1 & \cdots & 0\\
                \end{pmatrix}.
            \end{equation}
            \item $\bm{P}\in\left(\mathbb{R}^{2n}\right)^{n^2-n}$ is an $n^2-n\times 2n$ matrix that is arbitrary.
        \end{enumerate}
        It is obvious from this construction that $K$ is an $\frac{n\left(n-1\right)}{2}$-dimensional embedding of a compact subspace of hollow symmetric matrices (with bounded entries and norm) $\bm{B}$. We assume WLOG that the Jacobian of $\left.\mathcal{E}\right|_K$ (that is, the locally Lipschitz restriction of $\mathcal{E}$ to points of the form $\left(\bm{Q},\bm{0}\right)\in R$) attains its maximum rank at $\bm{B}=\bm{0}$ (that is, the squeezed state $\ket{\bm{0}}_{\bm{\hat{q}}}$); this can always be done by implicitly transforming the basis of the inputs to the model (i.e. by appropriately relabeling points in $\left(\mathbb{R}^{2n}\right)^{n^2}$), and then considering $K$ as previously defined in this new basis.
        
        We now give some intuition behind points in the smooth manifold (with boundary) $R$. At fixed $\bm{P}=\bm{0}$, the states described by the measurement scenarios of points in $R$ are exactly CV graph states without loops (with bounded weight edges, as $K$ is compact) and, depending on the measurement results, perhaps overall phases on the stabilizers. To see this, note that the symmetric constraint on $\bm{B}$ ensures that the symplectic product of any two rows of $\bm{Q}$ is zero; furthermore, the final $n$ columns of $\bm{Q}$ are linearly independent for all $\bm{B}\neq\bm{0}$, and the first $n$ columns for $\bm{B}=\bm{0}$ due to the shift by $\bm{H}$. Thus, all points in $K$ are full row rank, and the rows of $\bm{Q}$ completely determine the CV stabilizer state, up to phases given by the measurement results of these operators (which are not yet determined at the time of encoding). Furthermore, as $\bm{B}$ and $\bm{H}$ are hollow, the CV graph state $\bm{Q}$ describes has no loops. Also note that, up to trivial rescalings of the rows of $\bm{Q}$, different $\bm{Q}$ correspond to different graph states. At general $\bm{P}$, the state after the first $n$ measurements is still a CV graph state completely determined by $\bm{Q}$ (up to phases from the first $n$ measurement results); different $\bm{P}$ correspond to different (non-Gaussian) measurement scenarios given an initial CV graph state determined by $\bm{Q}$ (and the first $n$ measurement results).
        
        We will proceed as follows. First, we will show that as $\dim\left(L\right)<\frac{n\left(n-3\right)}{2}$, $\mathcal{E}$ must map three distinct CV graph states described by different $\bm{Q}$ to the same point in latent space when $\bm{P}=\bm{0}$. Then, we will use Lemma~\ref{lemma:graph_state_context} to show that the stabilizers of these states exhibit quantum contextuality (independent of the associated $n$ measurement results), and give rise to a distinguishing measurement sequence. Finally, we will show that one can locally find $\bm{P}\neq\bm{0}$ mapping to the same point in latent space that contains this distinguishing measurement sequence, forcing an incorrect measurement outcome on one of these states. This gives rise to an infinite backward empirical cross entropy on this task.
        
        Let us begin by showing that $\left.\mathcal{E}\right|_R$ (i.e. the locally Lipschitz map that is $\mathcal{E}$ restricted to $R$) must map three nontrivially distinct $\bm{Q}$ (i.e. three distinct CV graph states) to the same point in latent space. We will consider the restriction $\left.\mathcal{E}\right|_K$, which is the (locally Lipschitz) restriction of $\mathcal{E}$ to points of the form $\left(\bm{Q},\bm{0}\right)\in R$. We will show that this map is not injective for small enough $\dim\left(L\right)$. As $\left.\mathcal{E}\right|_R$ lifts to $\left.\mathcal{E}\right|_K$, this will show that three distinct $\bm{Q}$ map to the same point in latent space under $\left.\mathcal{E}\right|_R$. This is similar to the construction for $\left.\mathcal{F}\right|_K$ in the proof Theorem~\ref{thm:online_sep}; we repeat it here for completeness.
        
        By the constant rank theorem and the local Lipschitzness of $\left.\mathcal{E}\right|_K$, $\left.\mathcal{E}\right|_K$ is not injective for
        \begin{equation}
            \dim\left(L\right)<\dim\left(K\right)=\frac{n\left(n-1\right)}{2}.
        \end{equation}
        In particular, in a sufficiently small neighborhood of $\bm{B}=\bm{0}$ (where the Jacobian of $\left.\mathcal{E}\right|_K$ attains its maximal rank), there exist local coordinates $\bm{\tilde{x}}$ of $K$ and $L$ such that
        \begin{equation}
            \left.\mathcal{E}\right|_K\left(\tilde{x}_1,\ldots,\tilde{x}_{\frac{n\left(n-1\right)}{2}}\right)=\left(\tilde{x}_1,\ldots,\tilde{x}_l,0,\ldots,0\right)
        \end{equation}
        for some $l\leq\dim\left(L\right)<\frac{n\left(n-1\right)}{2}$~\cite{butler_timourian_viger_1988}. WLOG, we identify $\bm{\tilde{x}}=\bm{0}$ with $\bm{B}=\bm{0}$, which is the state infinitely squeezed in all $\hat{q}_i$. We will call $C$ the fiber with local coordinates
        \begin{equation}
            \bm{\tilde{x}}=\left(0,\ldots,0,\tilde{x}_{l+1},\ldots,\tilde{x}_{\frac{n\left(n-1\right)}{2}}\right),
        \end{equation}
        which is of dimension at least
        \begin{equation}
            \varDelta\equiv\frac{n\left(n-1\right)}{2}-\dim\left(L\right)\geq 1.
        \end{equation}
        By construction, all points in $C$---including $\bm{B}=\bm{0}$---map to the same point $l\in L$ under $\left.\mathcal{E}\right|_K$. We now assume that $\dim\left(L\right)<\frac{n\left(n-3\right)}{2}$ such that $\varDelta\geq n+1$.
        
        Now fix $\bm{B}=\bm{0}$ and $\bm{B'}\neq\bm{B}$ in $C$. As described previously, $\bm{Q}$ (and thus $\bm{B}$) completely determines a CV graph state after $n$ measurements, up to trivial rescalings of the rows of $\bm{Q}$ and the measurement results. As the dimension of the space that differ (modulo $\cpi$) from $\bm{B'}+\bm{H}/2$ by just a scaling factor in each row is at most $n$, because $\varDelta\geq n+1$ we must have that there exists another $\bm{B''}\neq\bm{B},\bm{B'}$ describing a distinct CV graph state. Therefore, by Lemma~\ref{lemma:graph_state_context}, we have that there exists a distinguishing measurement $s$ that is a stabilizer of the state corresponding to $\bm{B}=\bm{0}$. Note that as Lemma~\ref{lemma:graph_state_context} does not depend on the phases of the CV stabilizers, the existence of this distinguishing measurement holds true regardless of what the measurement results are (i.e. independently from what the model outputs for the first $n$ tokens in the decoded sequence). Depending on these measurement results, however, the distinguishing measurement could be one of three different measurement sequences, depending on the validity of Eq.~\eqref{eq:phase_rel_dist_meas}.
        
        We have now shown that there exists $\left(\bm{Q_i},\bm{0}\right)$ for $1\leq i\leq 3$ such that all $\bm{Q_i}$ are distinct, and that there exists a measurement of a stabilizer of $\bm{Q_1}$ such that the post-measurement states of $\bm{Q_2}$, $\bm{Q_3}$ are orthogonal. We will now show that there exist $\bm{P_i}$ such that $\left(\bm{Q_i},\bm{P_i}\right)$ also maps to the same point in latent space, and $\bm{P_i}$ is a distinguishing measurement sequence. This will give the final separation.
        
        First, note that one can find a $\bm{P_i}$ given an arbitrarily small bound on its norm that encodes a distinguishing measurement sequence; this is because, in the proof of Lemma~\ref{lemma:graph_state_context}, one can repeatedly take $\alpha\to 3\alpha$ to yield a distinguishing measurement sequence using arbitrarily large stabilizer powers, which corresponds to arbitrarily small $\bm{P_i}$ in the modified stabilizer measurement translation task. Now, consider $\left.\mathcal{E}\right|_S$, defined as the restriction of $\mathcal{E}$ to points of the form $\left(\bm{Q_i},\bm{P}\right)\in R$ for any of the fixed $\bm{Q_i}\in K$, where $\bm{P}$ is zero in all rows except for the last $S$ rows. Initially, let $S=S_1=2$; that is, $\bm{P}$ is restricted to be all zero except for its final two rows, which are allowed to vary. Let $d_{S_1}$ be the dimension of the image of $\left.\mathcal{E}\right|_S$. If $d_{S_1}=0$, then $\mathcal{E}$ is locally independent from the final two rows of $\bm{P}$ at $\bm{Q_i},\bm{m_i}$ when all other rows are fixed to be zero, and we are done as we can set these two rows to be the distinguishing measurement sequence. If $d_{S_1}>0$, then consider $S=S_2=4$. If $d_{S_2}=d_{S_1}$, then for all local choices of the third and fourth final rows, there exists a choice of the final two rows such that $\mathcal{E}$ is constant. Then, all of the possible distinguishing measurement scenarios can be encoded into the third and fourth final rows of $\bm{P}$ (with the appropriate choice of final two rows) and map to the same point in latent space as when $\bm{P}=\bm{0}$, yielding the appropriate distinguishing measurement sequence.
        
        If $d_{S_2}>d_{S_1}$ instead, we are able to iterate this procedure once more. As $\dim\left(L\right)<\frac{n\left(n-3\right)}{2}$, eventually this iteration will stop with $d_{S_i}=d_{S_{i-1}}$, and we have the freedom to set two rows of $\bm{P}$ to be the distinguishing measurement sequence and map to the same point in latent space as $\bm{P}=\bm{0}$.
        
        We have thus shown that the encoder-decoder model must map three points in $R$ that give rise to a distinguishing measurement scenario to the same point in latent space. That is, when $\dim\left(L\right)<\frac{n\left(n-3\right)}{2}$, there exists input sequences $\bm{s_i}$ that map to the same point in the latent space of the model that must give rise to orthogonal measurement results. As they are mapped to the same point in the latent space of the model, the model must output the same translation for all of them, giving at least one incorrect result. Thus, the backward empirical cross entropy when translating one of these $\bm{s_i}$ must be infinite.
    \end{proof}
    
    \section{Considerations for Experimental Implementations}\label{sec:exp_consds}
    
    We have shown in Appendix~\ref{sec:proof_of_express_sep} that CRNNs are more expressive than both online and encoder-decoder sequence models, assuming a locally Lipschitz condition on the models (and an additional technical condition on the latter class of models). Though CRNNs only utilize Gaussian operations---which are believed to be much simpler to implement that universal CV quantum computing~\cite{RevModPhys.77.513}---our model also utilizes non-Gaussian ancilla states to perform non-Gaussian measurements. As we formally compare our quantum model with infinite precision classical models, we have a formal requirement for GKP ancilla states with infinite homodyne precision measurements to show a separation.
    
    Of course, in practice, classical neural sequence models are finite precision. In particular, tensor processing unit (TPU) implementations of classical neural networks often use as imprecise as 8-bit arithmetic. We therefore expect that one can circumvent the formal need for GKP states to use \emph{qubit} ancilla states to perform phase estimation of the CV stabilizer operators to a precision matching that of classical neural networks. There are proposals for engineering the required longitudinal photon/qubit interactions using circuit quantum electrodynamics (QED)~\cite{Kerman_2013,RevModPhys.93.025005}. As similar couplings are already used in proposals for the generation of approximate GKP states~\cite{Shi_2019}, this direct approach is likely more experimentally feasible. Furthermore, such a finite precision implementation may actually \emph{gain} expressive power compared to infinite precision CRNNs. Assuming the back action of the finite precision measurement yields a finitely squeezed state in the CRNN, with this one can construct a model capable of universal CV quantum computation~\cite{PhysRevLett.123.200502,calcluth2022}. This also holds when the initial state of the model is taken to be the vacuum state (or any other finitely squeezed Gaussian state). This suggests a potential superpolynomial advantage in the expressive power and the time complexity of inference when this model is implemented at finite precision.
    
    If one wishes to avoid coupling to qubits, we make the bolder conjecture that any non-Gaussian ancilla state is enough to yield a separation. Our intuition for this comes from recent work~\cite{booth2021,haferkamp2021} demonstrating that non-Gaussian operations are equivalent to the presence of quantum contextuality. As the presence of quantum contextuality is the source of the separations in our proofs, this gives evidence that one could use a more experimentally feasible non-Gaussian ancilla state than a GKP state---such as a photon subtracted state~\cite{ra2020non}---and achieve similar results.
    
    \section{Classical Simulation of Gaussian Operations and GKP States}\label{sec:sim_deets}
    
    We now describe the high level strategy of classically simulating both Gaussian operations applied to Gaussian states, and (restricted) Gaussian operations applied to GKP states. The former strategy will roughly follow that given in \cite{PhysRevA.83.042335}, and the latter that given in \cite{calcuth2022}. These will be the building blocks of the Gaussian RNN and contextual RNN cells we numerically test in Sec.~\ref{sec:numerics}; we give the details of the full architectures, including the classical architectures, in Appendix~\ref{sec:numerics_deets}.
    
    \subsection{Gaussian States}\label{sec:gaussian_state_sim_deets}
    
    First, we describe the simulation of Gaussian operations performed on Gaussian states. It is well known that any $N$ mode Gaussian pure state can be created from the ground state of $N$ harmonic oscillators with unitary operations $\ce^{\ci f\left(\hat{\bm{q}},\hat{\bm{p}}\right)}$, where $f\left(\hat{\bm{q}},\hat{\bm{p}}\right)$ are at most quadratic in terms of the quadrature operators $\hat{\bm{q}}$ and $\hat{\bm{p}}$. By stacking $\hat{\bm{q}}$ and $\hat{\bm{p}}$ together, we call $\hat{\bm{x}}=\left(\hat{\bm{q}},\hat{\bm{p}}\right)^\intercal $ the quadrature operator. The Heisenberg evolution of the quadrature operator under a Gaussian unitary operator $\bm{R}$ is in general:
    \begin{equation}
        \hat{\bm{x'}}=\bm{R}^{\dagger}\hat{\bm{x}}\bm{R}=\bm{S}\hat{\bm{x}}+\bm{c},
    \end{equation}
    where $\bm{S}$ is a symplectic matrix of c-numbers, and $\bm{c}=\left(\bm{c}_{\bm{q}},\bm{c}_{\bm{p}}\right)^\intercal $ is a vector of c-numbers denoting the mode center shift. There are various ways to decompose a symplectic matrix $\bm{S}$; WLOG (as discussed in \cite{PhysRevA.83.042335}), we will consider symplectic matrices of the form:
    \begin{equation}
        \bm{S}=\begin{pmatrix}\bm{U}^{-1/2}& \bm{0}\\ \bm{V}\bm{U}^{-1/2} & \bm{U}^{1/2}\end{pmatrix},
    \end{equation}
    where $\bm{U}$ and $\bm{V}$ are both real symmetric matrices. In addition, $\bm{U}$ is positive definite, i.e. $\bm{U}=\bm{U}^\intercal >0$. If the covariance matrix of the $N$ mode ground state is 
    \begin{equation}
        \operatorname{cov}\left(\hat{\bm{x}}_0\right)=\frac{1}{2}\bm{I}, 
    \end{equation}
    then the covariance matrix of the transformed Gaussian state is:
    \begin{equation}
        \begin{aligned}
            \bm{\varSigma}&=\operatorname{cov}\left(\hat{\bm{x}}\right)=\operatorname{cov}\left(\bm{S}\hat{\bm{x}}_0\right)=\frac{1}{2}\left\langle\left\{\left(\bm{S}\hat{\bm{x}}_0\right)^{\dagger},\left(\bm{S}\hat{\bm{x}}_0\right)^\intercal\right\}\right\rangle\\
            &=\frac{1}{2}\bm{S}\bm{S}^\intercal =\frac{1}{2}\begin{pmatrix}\bm{U}^{-1}& \bm{U}^{-1}\bm{V}\\
            \bm{V}\bm{U}^{-1} & \bm{U}+\bm{V}\bm{U}^{-1}\bm{V}\end{pmatrix}.
        \end{aligned}
    \end{equation}
    We can also write down its wavefunction in position space as:
    \begin{equation}
    \begin{aligned}
        \psi_{\bm{Z},\bm{c}}\left(\bm{q}\right)&=\pi^{-N/4}\left(\det\left(\bm{U}\right)\right)^{1/4}\exp\left(-\frac{1}{2}\left(\bm{q}-\bm{c}_{\bm{q}}\right)^\intercal \left(\bm{U}-\ci\bm{V}\right)\left(\bm{q}-\bm{c}_{\bm{q}}\right)\right)\\
        &=\pi^{-N/4}\left(\det\left(\bm{U}\right)\right)^{1/4}\exp\left(\frac{\ci}{2}\left(\bm{q}-\bm{c}_{\bm{q}}\right)^\intercal \bm{Z}\left(\bm{q}-\bm{c}_{\bm{q}}\right)\right),
    \end{aligned}
    \end{equation}
    where $\bm{q}$ and $\bm{c}_{\bm{q}}$ are c-number column vectors, and $\bm{Z}=\bm{V}+\ci\bm{U}$ is a complex symmetric matrix. We can interpret $\bm{Z}$ as the adjacency matrix for an undirected graph with complex-valued edge weights. Therefore, any Gaussian pure states can be interpreted as a Gaussian graph states with complex-valued weights on the graph edge. As Gaussian states are uniquely identified by the linear combinations of position and momentum operators that nullify them~\cite{PhysRevLett.88.097904}, we can consider what the nullifiers are for the CV graph state defined by $\bm{Z}$ and $\bm{c}=\left(\bm{c}_{\bm{q}},\bm{c}_{\bm{p}}\right)^\intercal$:
    \begin{equation}
        \begin{aligned}
            \left(\hat{\bm{p}}-\bm{Z}\hat{\bm{q}}+\bm{Z}\bm{c}_{\bm{q}}-\bm{c}_{\bm{p}}\right)\ket{\psi_{\bm{Z},\bm{c}}}&= \left(\hat{\bm{p}}-\bm{Z}\hat{\bm{q}}+\bm{Z}\bm{c}_{\bm{q}}-\bm{c}_{\bm{p}}\right)\bm{R}_{\bm{Z},\bm{c}}\ket{0}\\ &= \bm{R}_{\bm{Z},\bm{c}}\bm{R}_{\bm{Z},\bm{c}}^{\dagger}\left(\hat{\bm{p}}-\bm{Z}\hat{\bm{q}}+\bm{Z}\bm{c}_{\bm{q}}-\bm{c}_{\bm{p}}\right)\bm{R}_{\bm{Z},\bm{c}}\ket{0}\\
            &=\bm{R}_{\bm{Z},\bm{c}}\begin{pmatrix}-\bm{Z} & \bm{I}\end{pmatrix}\left(\bm{S}_{\bm{Z}}\hat{\bm{x}}+\bm{c}\right)\ket{0}+\left(\bm{Z}\bm{c}_{\bm{q}}-\bm{c}_{\bm{p}}\right)\bm{R}_{\bm{Z},\bm{c}}\ket{0}\\
            &=\bm{R}_{\bm{Z},\bm{c}}\begin{pmatrix}-\bm{Z} & \bm{I}\end{pmatrix}\begin{pmatrix}\bm{U}^{-1/2}\hat{\bm{q}}+\bm{c}_{\bm{q}}\\ \bm{V}\bm{U}^{-1/2}\hat{\bm{q}}+\bm{U}^{1/2}\hat{\bm{p}}+\bm{c}_{\bm{p}}\end{pmatrix}\ket{0}+\left(\bm{Z}\bm{c_q}-\bm{c_p}\right)\bm{R}_{\bm{Z},\bm{c}}\ket{0}\\
            &=\bm{R}_{\bm{Z},\bm{c}}\left(-\ci\bm{U}^{1/2}\hat{\bm{q}}+\bm{U}^{1/2}\hat{\bm{p}}-\bm{Z}\bm{c}_{\bm{q}}+\bm{c_p}\right)\ket{0}+\left(\bm{Z}\bm{c}_{\bm{q}}-\bm{c}_{\bm{p}}\right)\bm{R}_{\bm{Z},\bm{c}}\ket{0}\\
            &=\bm{R}_{\bm{Z},\bm{c}}\left(-\bm{Z}\bm{c_q}+\bm{c_p}\right)\ket{0}+\left(\bm{Z}\bm{c}_{\bm{q}}-\bm{c}_{\bm{p}}\right)\bm{R}_{\bm{Z},\bm{c}}\ket{0}\\
            &=\left(-\bm{Z}\bm{c}_{\bm{q}}+\bm{c_p}\right)\bm{R}_{\bm{Z},\bm{c}}\ket{0}+\left(\bm{Z}\bm{c}_{\bm{q}}-\bm{c}_{\bm{p}}\right)\bm{R}_{\bm{Z},\bm{c}}\ket{0}=\bm{0}.
        \end{aligned}
    \end{equation}
    In the second to last line, we have use the fact that $-\bm{Z}\bm{c}_{\bm{q}}+\bm{c}_{p}$ is a c-number vector, which commutes with $\bm{R}_{\bm{Z},\bm{c}}$. In addition, we also use the fact that $\bm{Z}=\bm{V}+\ci\bm{U}$, and $\left(-\ci\hat{\bm{q}}+\hat{\bm{p}}\right)\ket{0}=-\ci\sqrt{2}\hat{\bm{a}}\ket{0}=0$. Therefore, the Gaussian graph state with complex adjacency matrix $\bm{Z}$ and mode shift center $\bm{c}$ has complex nullifiers: 
    \begin{equation}
        \left(\hat{\bm{p}}-\bm{c}_{\bm{p}}-\bm{Z}\left(\hat{\bm{q}}-\bm{c}_{\bm{q}}\right)\right)\ket{\psi_{\bm{Z},\bm{c}}}=\bm{0}.
    \end{equation}
    
    We now restrict to the case $\bm{V}=\bm{0}$ and $\bm{c}_{\bm{p}}=\bm{0}$, which are the class of states we consider in our numerical experiments. Then, the nullifiers are of the form: 
    \begin{equation}
        \left(\hat{\bm{p}}-\ci\bm{U}\hat{\bm{q}}+\ci\bm{U}\bm{c}_{\bm{q}}\right)\ket{\psi_{\bm{Z},\bm{c}}}=0.
    \end{equation}
    We also restrict to symplectic transformations of the quadratures of the form:
    \begin{equation}
        \bm{S}=\begin{pmatrix}\bm{W}^\intercal &\bm{0}\\\bm{0}&\bm{W}^{-1}\end{pmatrix},\label{eq:s_rel}
    \end{equation}
    where $\bm{W}$ is some assumed invertible matrix. We restrict to operations of this form to more efficiently allow for the classical simulation of the contextual RNN using Gaussian operations of an identical form, as discussed in \cite{calcuth2022} and Appendix~\ref{sec:contextual_sim_deets}. After performing the quantum operation described by such a symplectic matrix, the nullifier for the whole system is updated as $\left(\bm{W}^{-1}\hat{\bm{p}}-\ci\bm{U}\bm{W}^\intercal \hat{\bm{q}}+\ci\bm{U}\bm{c_q}\right)\ket{\psi}=\bm{0}$, which is equivalent to $\left(\hat{\bm{p}}-\ci\bm{W}\bm{U}\bm{W}^\intercal \hat{\bm{q}}+\ci\bm{W}\bm{U}\bm{c_q}\right)\ket{\psi}=\bm{0}$. Homodyne detection of the position quadrature on $m$ qumodes is then, in general, a multivariate Gaussian random variable which is centered at $\bm{\varPi_Y}\bm{W}^{-1\intercal}\bm{c_q}$ with variance $\bm{\varPi_Y}\bm{U}^{-1}\bm{\varPi_Y}^\intercal$, where $\bm{\varPi_Y}$ is the projection operator onto the subspace of the $m$ qumodes being measured. To make the training of our Gaussian models more stable---and to maintain simulability with GKP input states, as is done in Appendix~\ref{sec:contextual_sim_deets}---we assume in our numerical simulations that there is an implicit large scaling for $\bm{U},\bm{c_q}$ such that this variance is small. After the measurement, the hidden state at $i$th step is updated to a generalized graph state with adjacency matrix $\bm{\varPi_H}\bm{U}\bm{\varPi_H}^\intercal$ and mode shift $\bm{\varPi_H}\bm{W}^{-1\intercal}\bm{c_q}$, where $\bm{\varPi_H}$ is the projection operator onto the subspace of the latent $n$ qumodes.
    
    \subsection{Gaussian Operations on GKP States}\label{sec:contextual_sim_deets}
    
    Our methods for the simulation of (restricted) Gaussian operations on GKP states are similar to the methods used by \cite{calcuth2022}. We restrict to (unnormalized) states of the form:
    \begin{equation}
        \ket{\psi}=\sum\limits_{\bm{\ell}\in\mathbb{Z}^n}\ket{\psi_{\bm{\ell}}},
    \end{equation}
    where each $\ket{\psi_{\bm{\ell}}}$ is a Gaussian state with nullifiers given by:
    \begin{equation}
        \bm{\hat{n}_{\bm{\ell}}}=\epsilon\bm{\hat{p}}-\bm{Z}\left(\bm{\hat{q}}-\bm{c_q}-\bm{L}\bm{\ell}\right).
    \end{equation}
    We assume $\epsilon\to 0^+$, such that all $\ket{\psi_{\bm{\ell}}}$ are approximately orthogonal. Time evolution is simulated as in Appendix~\ref{sec:gaussian_state_sim_deets}, simultaneously for all $\ket{\psi_{\bm{\ell}}}$. As all $\ket{\psi_{\bm{\ell}}}$ are approximately orthogonal, measurement is simulated via choosing $\bm{\ell}$ uniformly at random, and performing the corresponding measurement. In principle, using many measurements (over multiple instances of the state), one can read out $\bm{L}$ and $\bm{c_q}$; these are the measurement results we use in our numerical expirements, as described in Appendix~\ref{sec:quantum_nums_deets}. For any given measurement on a subset of the modes, the post-measurement state on the remainder of the modes is just the uniform superposition over all $\bm{\ell}$ consistent with the resulting measurement outcome $\bm{y}$.
    
    To make this latter observation more concrete, assume after evolution under $\bm{S}$ of the form of Eq.~\eqref{eq:s_rel}, the nullifiers of $\ket{\psi_{\bm{\ell}}}$ are (see Appendix~\ref{sec:numerics_deets}):
    \begin{equation}
        \bm{\hat{n}_{\bm{\ell}}}=\epsilon\bm{\hat{p}}-\bm{W}\bm{Z}\bm{W}^\intercal\left(\bm{\hat{q}}-\bm{W}^{-1\intercal}\bm{c_q}-\bm{W}^{-1\intercal}\bm{L}\bm{\ell}\right),
    \end{equation}
    We wish to find all $\bm{\ell}$ consistent with the position measurement result $\bm{y}$ (in the limit $\epsilon\to 0^+$); that is, all $\bm{\ell}$ such that:
    \begin{equation}
        \bm{\varPi_Y}\left(\bm{W}^{-1\intercal}\bm{c_q}+\bm{W}^{-1\intercal}\bm{L}\bm{\ell}\right)=\bm{y},
        \label{eq:ell_const}
    \end{equation}
    where $\bm{\varPi_Y}$ is the projector onto the $m$ mode space being measured. Let $H$ label the space of the other $n$ modes, and the projector onto this space $\bm{\varPi_H}$. Writing (with the assumptions that $\bm{\tilde{W}_{YY}}$ and $\bm{W_{HH}}$ are full rank):
    \begin{align}
        \bm{W}^{-1\intercal}&=\begin{pmatrix}
        \bm{\tilde{W}_{HH}} & \bm{\tilde{W}_{HY}}\\
        \bm{\tilde{W}_{YH}} & \bm{\tilde{W}_{YY}}
        \end{pmatrix},\\
        \bm{W}^\intercal&=\begin{pmatrix}
        \bm{W_{HH}} & \bm{W_{HY}}\\
        \bm{W_{YH}} & \bm{W_{YY}}
        \end{pmatrix},\\
        \bm{\ell}&=\begin{pmatrix}
        \bm{\ell_H}\\
        \bm{\ell_Y}
        \end{pmatrix},
    \end{align}
    and assuming $\bm{L}$ (assumed full rank) is of the form:
    \begin{equation}
        \bm{L}=\begin{pmatrix}
        \bm{L_{HH}} & 0\\
        0 & \bm{L_{YY}}
        \end{pmatrix},
    \end{equation}
    we find from Eq.~\eqref{eq:ell_const} that all $\bm{\ell}$ consistent with the measurement result satisfy:
    \begin{equation}
        \bm{L_{YY}}\bm{\ell_Y}=\bm{\tilde{W}_{YY}}^{-1}\left(\bm{y}-\bm{\varPi_Y}\bm{W}^{-1\intercal}\bm{c_q}-\bm{\tilde{W}_{YH}}\bm{L_{HH}}\bm{\ell_H}\right).
    \end{equation}
    Assuming the entries of $\bm{L_{YY}}$ are sufficiently small, up to any given machine precision the $\bm{\ell_H}$ are in one-to-one correspondence with the $\bm{\ell}$ consistent with the measurement result. Furthermore, for all such $\bm{\ell}$:
    \begin{equation}
        \begin{aligned}
            \bm{\varPi_H}\left(\bm{W}^{-1\intercal}\bm{c_q}+\bm{W}^{-1\intercal}\bm{L}\bm{\ell}\right)=&\bm{\varPi_H}\bm{W}^{-1\intercal}\bm{c_q}+\bm{\tilde{W}_{HH}}\bm{L_{HH}}\bm{\ell_H}+\bm{\tilde{W}_{HY}}\bm{L_{YY}}\bm{\ell_Y}\\
            =&\bm{\varPi_H}\bm{W}^{-1\intercal}\bm{c_q}+\bm{\tilde{W}_{HH}}\bm{L_{HH}}\bm{\ell_H}+\bm{\tilde{W}_{HY}}\bm{\tilde{W}_{YY}}^{-1}\left(\bm{y}-\bm{c_q}-\bm{\tilde{W}_{YH}}\bm{L_{HH}}\bm{\ell_H}\right)\\
            =&\bm{\varPi_H}\bm{W}^{-1\intercal}\bm{c_q}+\left(\bm{\tilde{W}_{HH}}-\bm{\tilde{W}_{HY}}\bm{\tilde{W}_{YY}}^{-1}\bm{\tilde{W}_{YH}}\right)\bm{L_{HH}}\bm{\ell_H}\\
            &+\bm{\tilde{W}_{HY}}\bm{\tilde{W}_{YY}}^{-1}\left(\bm{y}-\bm{\varPi_Y}\bm{W}^{-1\intercal}\bm{c_q}\right)\\
            =&\bm{\varPi_H}\bm{W}^{-1\intercal}\bm{c_q}+\bm{W_{HH}}^{-1}\bm{L_{HH}}\bm{\ell_H}+\bm{\tilde{W}_{HY}}\bm{\tilde{W}_{YY}}^{-1}\left(\bm{y}-\bm{\varPi_Y}\bm{W}^{-1\intercal}\bm{c_q}\right).
            \label{eq:post_meas_centers}
        \end{aligned}
    \end{equation}
    Using the appropriate displacement operator to remove the final term of Eq.~\eqref{eq:post_meas_centers}, then, yields the effective transformation:
    \begin{align}
        \bm{c_q}&\mapsto\bm{\varPi_H}\bm{W}^{-1\intercal}\bm{c_q},\\
        \bm{L}&\mapsto\bm{W_{HH}}^{-1}\bm{L_{HH}}.
    \end{align}
    
    \section{Details of the Numerical Simulations}\label{sec:numerics_deets}
    
    \begin{figure}
        \begin{center}
            \includegraphics[width=0.75\linewidth]{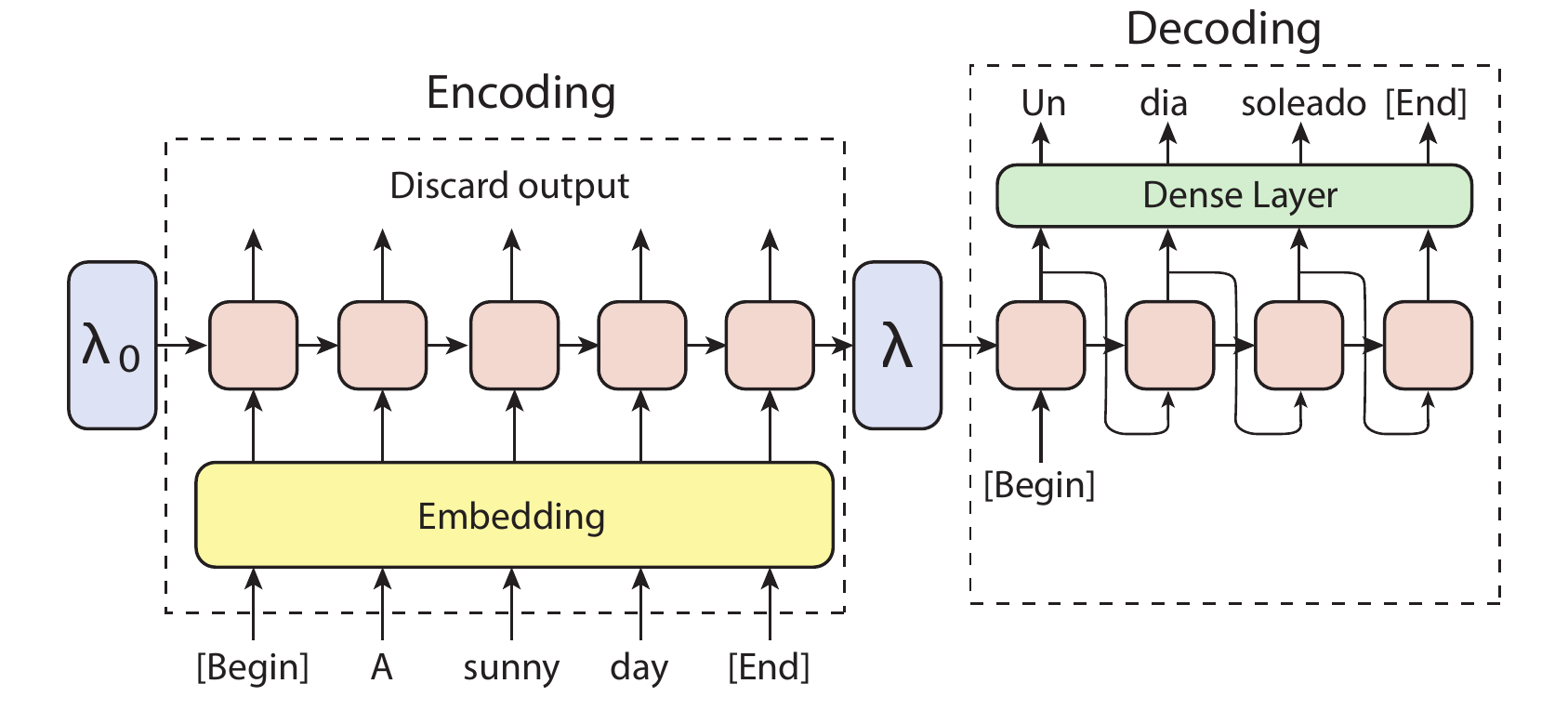}
            \caption{An overview of the recurrent models we study. Each red box represents the recurrent cell, and the variational parameters in the recurrent cells are shared within the encoder and decoder. $\bm{\lambda_0}$ is a random fixed hidden memory vector. In the decoding phase, the output of each recurrent cell is also treated as the input for the next recurrent cell.\label{fig:online_architecture}}
        \end{center}
    \end{figure}
    \begin{figure}
        \begin{center}
            \includegraphics[width=0.45\linewidth]{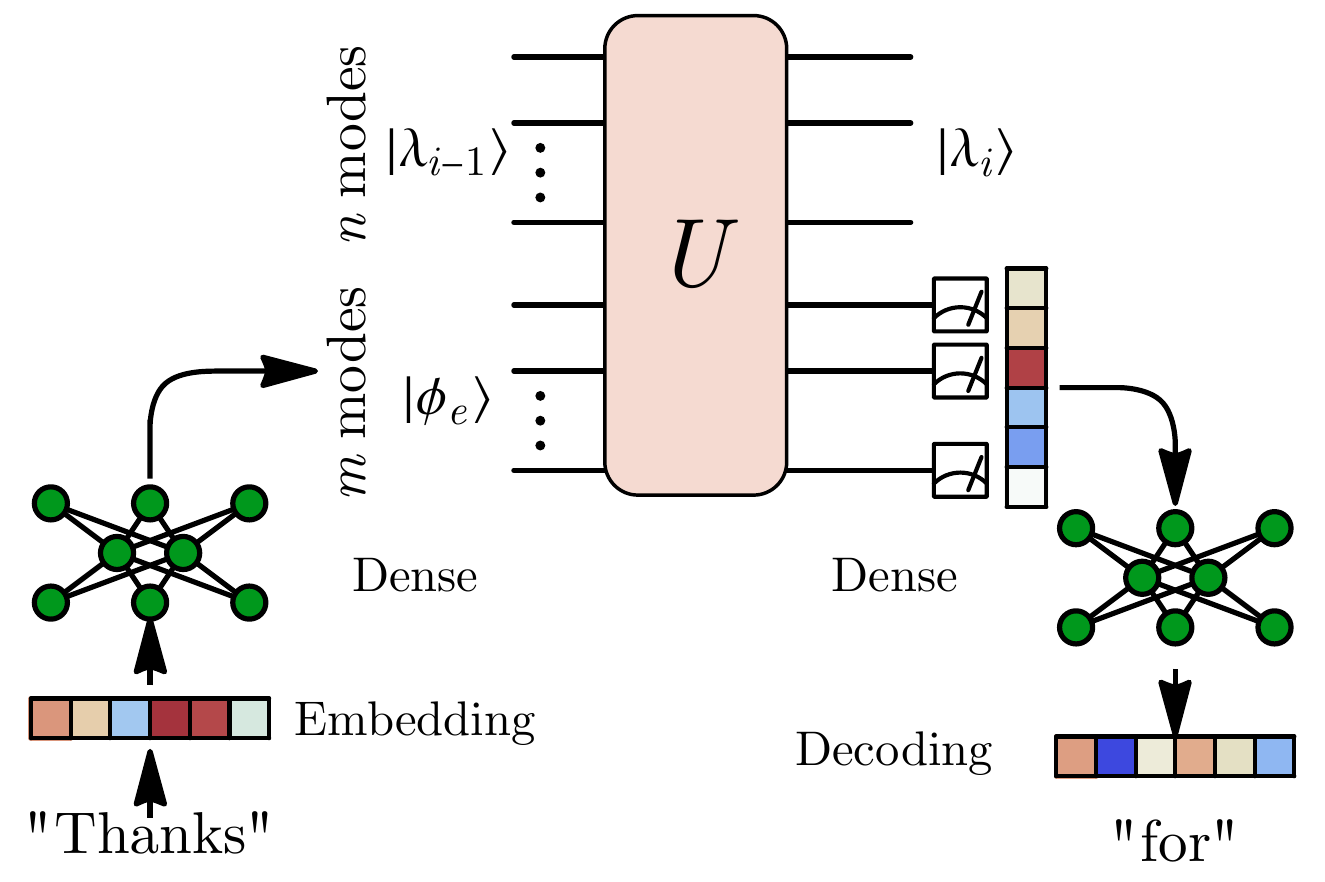}
            \caption{One recurrent cell of the quantum recurrent architectures. Note that the only trained part of the dense network at the input of each recurrent cell are displacements in phase space acting on $\ket{\phi_e}$, in order to keep the number of trainable parameters in line with the GRU RNN at the same $n$. $\ket{\phi_e}$ is a Gaussian state for the Gaussian model, and a GKP state for the CRNN.\label{fig:quantumcell}}
        \end{center}
    \end{figure}
    We now discuss the details of the numerical simulations performed in Sec.~\ref{sec:numerics}. For all models, we studied the performance in modeling a standard Spanish-to-English data set~\cite{spaengdata}. For each model and each $n$, the training set was taken to be a random sample of 80\% of the data set, and the test set 20\%. Each model was trained for 80 epochs, with a batch size of 64. To map the words in this data set to vectors (taken to also be of dimension $n$, the model dimension), we used the Keras~\cite{chollet2015keras} implementation of word2vec (``TextVectorization'') adapted to the data, with a maximum vocabulary size of 5000. The first 5000 most frequent words are mapped to distinct integers, and other words are mapped to a unique token ``[Unk].'' At the beginning and end of each sentence, we add unique ``[Begin]'' and ``[End]'' tokens. For the recurrent translation models, such as the GRU RNN or the CRNN, we do not need to set the length of the sentences. For the Transformer model, we set the sentence length to be 20 words. If the sentence is shorter than 20 words, the additional token ``[Pad]'' is added to the sentence. For the rare case when the sentence contains more than 20 words, additional words are truncated. All networks were trained using Adam~\cite{DBLP:journals/corr/KingmaB14} (with a learning rate of $10^{-3}$, $\beta_1=0.9$, $\beta_2=0.999$, and $\epsilon=10^{-7}$), trained on the forward empirical cross entropy (as the backward empirical cross entropy is difficult to train on).

    \subsection{Classical Sequence Models}\label{sec:classical_numerics_deets}
    
    We studied three standard classical sequence models in our numerical experiments: an implementation of an orthogonal recurrent neural network~\cite{pmlr-v70-jing17a}, an implementation of a network using gated recurrent units (GRU)~\cite{cho-etal-2014-learning}, and an implementation of a Transformer~\cite{10.5555/3295222.3295349}. The first two networks were trained in a seq2seq configuration~\cite{10.5555/2969033.2969173}; the models autoregressively map the input sequence to a latent space, which then is autoregressively decoded. For the orthogonal recurrent neural network, we used the implementation of \cite{pmlr-v70-jing17a}, with network capacity equal to the model dimension. An illustration of these architectures is given in Fig.~\ref{fig:online_architecture}.
    
    The Transformer models we considered follow the standard construction of \cite{10.5555/3295222.3295349}. To fairly compare against the shallow RNN cells we consider, we used a single Transformer encoder and decoder layer for each model. Our implementation used a trained positional embedding with uniform initialization, and the encoders and decoders used ReLU activations in the feedforward network layers, Glorot weight initialization, zero bias initialization, no dropout, and a single head. Each encoder and decoder layer was followed by the layer normalization implementation of Keras~\cite{chollet2015keras} with its default parameters. The final layer normalization of the decoder of each Transformer we considered was followed by a dense layer with a softmax activation function.
    \begin{algorithm}
        \KwIn{\tcp*[h]{cell inputs}\newline$n\times n$ latent graph adjacency matrix $\bm{A_{i-1}}$\newline$n\times n$ lattice $\bm{J_{i-1}}$\newline $n\times 1$ stabilizer phases $\bm{\alpha_{i-1}}$\newline $m\times 1$ input $\bm{x_i}$\newline\newline\tcp*[h]{trainable parameters}\newline$\left(n+m\right)\times\left(n+m\right)$ weight matrix $\bm{W}$\newline Dense layers $\bm{f},\bm{g}$\newline\newline\tcp*[h]{constants}\newline Dense layers $\bm{h},\bm{r}$\newline Projectors onto latent and input spaces, respectively, $\bm{\varPi_H},\bm{\varPi_Y}$}
        \KwOut{$n\times n$ latent graph adjacency matrix $\bm{A_i}$\newline $n\times n$ lattice $\bm{J_i}$\newline$n\times 1$ stabilizer phases $\bm{\alpha_i}$\newline $m\times 1$ measurement outcome $\bm{y}$}
        \Begin {
            $\bm{\alpha}\leftarrow\bm{\alpha_{i-1}}+\bm{A_{i-1}}^{-1}\bm{f}\left(\bm{x_i}\right)$\tcp*{perform mode shifts using a general function $f$}
            \BlankLine
            \BlankLine
            $\bm{\beta}\leftarrow\bm{g}\left(\bm{x_i}\right)$\tcp*{prepare the state associated with the input register}
            $\bm{K}\leftarrow\bm{h}\left(\bm{x_i}\right)$\;
            $\bm{S}\leftarrow\bm{r}\left(\bm{x_i}\right)$\;
            $\bm{B}=\bm{S}\bm{S}^\intercal$\tcp*{ensure the adjacency matrix is positive semidefinite}
            \BlankLine
            \BlankLine
            $\bm{U}\leftarrow\bm{A_{i-1}}\oplus\bm{B}$\tcp*{consider the tensor product of the latent and input states}
            $\bm{\gamma}\leftarrow\bm{\alpha}\oplus\bm{\beta}$\;
            $\bm{L}\leftarrow\bm{J_{i-1}}\oplus\bm{K}$\;
            \BlankLine
            \BlankLine
            $\bm{U}\leftarrow\bm{W}\bm{U}\bm{W}^\intercal$\tcp*{transform the graph state by performing a Gaussian operation}
            $\bm{\gamma}\leftarrow\bm{W}^{-1\intercal}\bm{\gamma}$\;
            $\bm{L}\leftarrow\bm{W}^{-1\intercal}\bm{L}$\;
            \BlankLine
            \BlankLine
            $\bm{y}\leftarrow\bm{\varPi_Y}\bm{L},\bm{\varPi_Y}\bm{\gamma}$\tcp*{read out the lattice and stabilizer phases}
            \BlankLine
            \BlankLine
            $\bm{A_i}\leftarrow\bm{\varPi_H}\bm{U}\bm{\varPi_H}^\intercal$\tcp*{project out the measured register}
            $\bm{J_i}\leftarrow\left(\bm{\varPi_H}\bm{W}\bm{\varPi_H}^\intercal\right)^{-1}\bm{J_{i-1}}$\;
            $\bm{\alpha_i}\leftarrow\bm{\varPi_H}\bm{\gamma}$\;
        }
        \caption{Contextual RNN Cell\label{alg:context_sim_pseudo}}
    \end{algorithm}
    
    \subsection{Quantum Sequence Models}\label{sec:quantum_nums_deets}
    
    Based on the discussion in Appendix~\ref{sec:sim_deets}, we simulated both the Gaussian RNN and the contextual RNN described in Sec.~\ref{sec:numerics}. The training and architecture of these models were identical to those of the orthogonal neural network described in Appendix~\ref{sec:classical_numerics_deets}, other than the structure of each unit cell of the recurrent network. Once again, Fig.~\ref{fig:online_architecture} describes the overall architecture of the models, and Fig.~\ref{fig:quantumcell} describes the recurrent cell of these models. For the Gaussian model, the simulated homodyne position measurements are what we used for our cell outputs; for the CRNN, we simulated the lattice (and position measurement) readout procedure described at the end of Appendix~\ref{sec:contextual_sim_deets}. The pseudocode for the cells of both are given in Algorithm~\ref{alg:context_sim_pseudo}; for the Gaussian model, one takes $\bm{J_0},\bm{K}=\bm{0}$. By considering the Gaussian RNN cell as a limit of the contextual RNN cell, and by not training $\bm{K}$ or $\bm{J_0}$ in the contextual RNN, we were able to maintain identical parameter counts for the two models. For the contextual RNN, $\bm{J_0}$ is fixed to the identity, and $\bm{K}$ is the result of an untrained dense layer $\bm{h}$ applied to the cell input (with uniform Glorot initialization and linear activation function, and biased such that $\bm{K}$ has mean the identity). Specializing to the notation of Algorithm~\ref{alg:context_sim_pseudo}: $\bm{r}$ is similar. $\bm{f}$ and $\bm{g}$ are similar, except with no bias.
    
    \subsection{Time Complexity}\label{sec:time_complexity}
    
    We now discuss the time complexity of implementing a CRNN, both as a quantum model implemented on a quantum computer, and as a quantum-inspired classical algorithm. On a quantum computer, each cell of the CRNN we consider in our proofs of an expressivity separation can be implemented in depth $\operatorname{O}\left(n\right)$, assuming access to the fixed ancilla state $\ket{a}$ used for non-Gaussian measurement. This further decreases to $\operatorname{O}\left(1\right)$ time if one assumes access to quantum fan-out~\cite{v001a005}. More general Gaussian operations can also be implemented in depth $\operatorname{O}\left(n\right)$ utilizing a swap network~\cite{PhysRevLett.120.110501}.
    
    Examining Algorithm~\ref{alg:context_sim_pseudo}, it is easy to see that our algorithm simulates inference on a CRNN with model dimension $n$ with time complexity $\operatorname{O}\left(n^{\upomega}\right)$. Here, $\upomega$ is the matrix multiplication exponent, with best-known bounds $2\leq\upomega<2.37286$~\cite{doi:10.1137/1.9781611976465.32}. Our results show that on certain tasks, such a CRNN of model dimension $n$ performs on par with e.g. GRU RNNs with model dimension $\operatorname{\Omega}\left(n^2\right)$, which performs inference in time $\operatorname{\Omega}\left(n^4\right)$ due to matrix-vector multiplications present in the model~\cite{cho-etal-2014-learning}. Thus, our classical simulation of CRNNs may be thought of as a quantum-inspired classical model that, though it is not efficient as implementing a CRNN on a quantum computer, is asymptotically more time efficient in inference and training than typical RNNs with an $n^2$-dimensional latent space. Of course, our Algorithm~\ref{alg:context_sim_pseudo} relies on matrix inversion. Though asymptotically matrix inversion takes time $\operatorname{O}\left(n^{\upomega}\right)$, unlike matrix multiplication it has poor GPU implementations and thus often is slow in practice. We leave further investigation of the practical utility of these quantum-inspired classical models to future work.
    
    \section{Supplementary Numerical Results}\label{sec:sup_nums}
    
    We now provide supplemental numerical experiments, comparing CRNNs with Transformers~\cite{10.5555/3295222.3295349} and a formulation of linear RNNs dubbed efficient unitary neural networks (EUNNs)~\cite{pmlr-v70-jing17a}.
    
    \begin{figure}
        \begin{center}
            \includegraphics[width=0.5\textwidth]{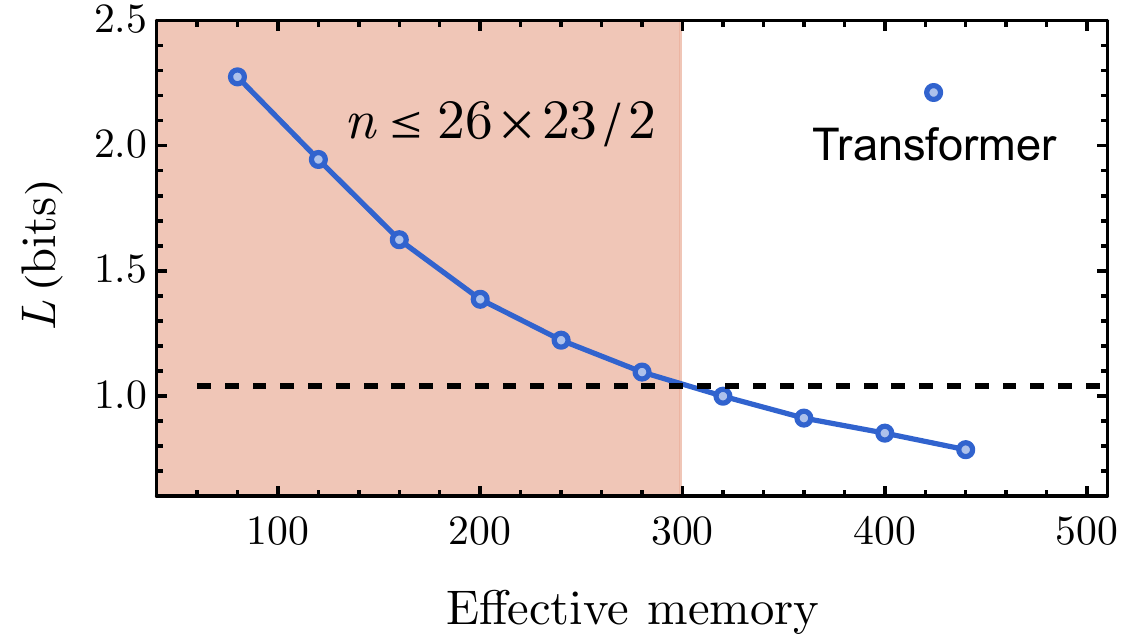}
            \caption{The performances of the $n=26$ CRNN model and the Transformer models. The dashed line shows the converged forward empirical cross entropy ($L$) for the $n=26$ CRNN model. The dimension of the latent space (labeled ``Effective memory'') of a Transformer is the model dimension multiplied by the length of the (padded) input sentences. The red region labeled ``$n\leq 26\times 23/2$'' is where the dimension of the memory of the Transformer is at most $\frac{n\left(n-3\right)}{2}$, where $n$ is the model dimension of the CRNN.\label{fig:trans_results}}
        \end{center}
    \end{figure}
    \begin{figure}
        \begin{center}
            \includegraphics[width=0.5\textwidth]{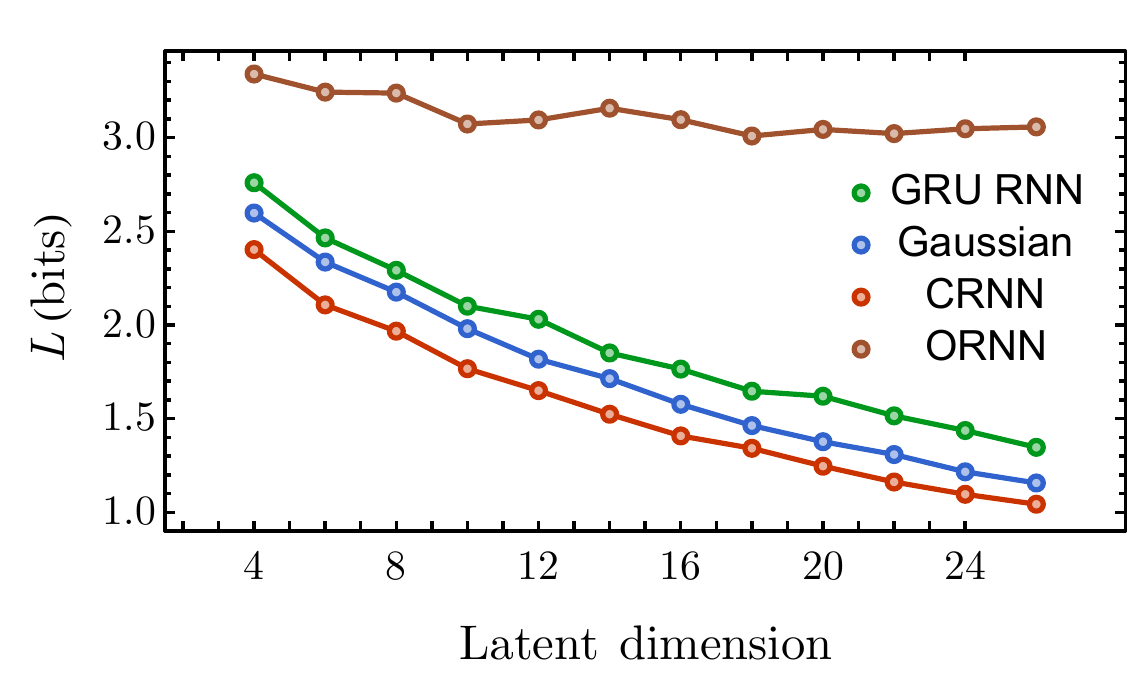}
            \caption{The converged forward empirical cross entropy ($L$) as a function of the model dimensions $n$ for ORNNs, and the online models we considered in Sec.~\ref{sec:numerics}. We see that ORNNs are greatly outperformed by the other online models we consider at the given task.\label{fig:ornn_results}}
        \end{center}
    \end{figure}
    The difficulty in comparing CRNNs and Transformers is that the effective memory of a Transformer---in the language of Fig.~\ref{fig:encoder_decoder_model}(b), the dimension $n$ of the latent space of the model---grows with the sentence length. Thus, we fixed a trained $n=26$ CRNN, and compare the performance of a Transformer at a variety of model dimensions against this model. We plot these results in Fig.~\ref{fig:trans_results}. As a guide to the eye, we also plot where the effective memory of the Transformer is at most $\frac{n\left(n-3\right)}{2}$, where $n$ is the model dimension of the CRNN.
    
    We also consider the performance of EUNNs compared with CRNNs, as both are linear models. We constrain the EUNN to be real---as in our simulations of Gaussian models and CRNNs---and call the resulting model an \emph{orthogonal recurrent neural network} (ORNN), using the implementation from \cite{pmlr-v70-jing17a}. We see in Fig.~\ref{fig:ornn_results} that CRNNs---and indeed, all of the online models we consider---greatly outperformed ORNNs at a variety of model dimensions.
\end{spacing}
\end{document}